\documentclass[11pt]{article}

\usepackage[margin=1in]{geometry}
\usepackage[T1]{fontenc}
\usepackage[utf8]{inputenc}
\usepackage{microtype}
\usepackage{amsmath, amssymb, amsthm, mathtools}
\usepackage{xcolor}
\usepackage{booktabs}
\usepackage{multirow}
\usepackage{graphicx}
\usepackage{tikz}
\usetikzlibrary{arrows.meta,positioning,calc}
\usepackage{enumitem}
\usepackage[numbers,sort&compress]{natbib}
\usepackage{algorithm}
\usepackage{algorithmicx}
\usepackage{algpseudocode}
\usepackage{pifont}
\usepackage{placeins}
\usepackage[colorlinks=true,linkcolor=blue!60!black,citecolor=blue!60!black,urlcolor=blue!60!black]{hyperref}
\graphicspath{{figures/}}

\newtheorem{lemma}{Lemma}
\newtheorem{theorem}{Theorem}
\newtheorem{corollary}{Corollary}
\newtheorem{proposition}{Proposition}

\newcommand{\eps}{\varepsilon}
\newcommand{\Dhat}{\widehat{D}}
\newcommand{\Ahat}{\widehat{A}}
\newcommand{\tsupp}{\tau_{\mathrm{supp}}}
\newcommand{\defeq}{\mathrel{:=}}
\newcommand{\cmark}{\ding{51}}
\newcommand{\xmark}{\ding{55}}

\title{Efficient Nash Equilibrium Computation for Cybersecurity Games}
\author{Michael Lanier \qquad David Farmer \qquad Yevgeniy Vorobeychik\\[4pt]
Washington University in St.\ Louis}
\date{}

\begin{document}
\maketitle

\begin{abstract}
Game-theoretic analyses of cyber defence often compute equilibria of games whose payoffs
exist only as the output of a simulator. Iterative equilibrium-finding methods grow a set of
attacker and defender policies and need the payoff of every attacker--defender pair, so
they are bottlenecked by payoff estimation: each payoff costs many simulator runs. We introduce Regret-Weighted Payoff Sampling
(RWPS), which spends a fixed simulation budget on the payoffs the equilibrium actually
depends on and predicts the rest with a model trained on every payoff measured so far.
Standard error bounds for estimated games are driven by the worst-estimated payoff, so they
cannot credit an estimator that is inaccurate only where accuracy does not matter. We prove a
bound that weights payoff errors by the opponent's equilibrium strategy, a certificate that
can be computed from simulated payoffs alone, and a condition under which errors in the
predicted payoffs cannot change either player's regret. On three synthetic general-sum
games, one of them a Colonel Blotto game of military resource allocation, the new bounds are
four to six times tighter than the standard one, and RWPS finds less exploitable equilibria
than minimum-regret-first search, information-gain search and progressive sampling at the
same budget. On two cyber-defence simulators, CyGym and a new game whose hosts are LLM agents
exposed to prompt injection, it gives the least exploitable equilibria at the smallest
budgets.
\end{abstract}

\section{Introduction}
\label{sec:intro}

Defending a network is a contest between two adaptive parties. An attacker probes for
weaknesses and changes tactics when blocked; a defender allocates limited monitoring and
response effort without seeing everything the attacker does. Game theory gives a principled
way to plan against such an opponent: model the interaction as a game, compute an
equilibrium, and deploy the defender's equilibrium strategy, which no attacker strategy can
exploit by more than the equilibrium's exploitability. For realistic networks the game cannot
be written down in closed form. What a pair of strategies is worth to each side is known only
by simulating the network under both of them, which makes these \emph{simulation-based
security games}. A family of simulators now supports this work. The Cyber Operations
Research Gym~\citep{standen2021cyborg,cage2023} underpins the CAGE challenges, Yawning
Titan~\citep{andrew2022yawningtitan}, CyberBattleSim~\citep{cyberbattlesim2021} and
FARLAND~\citep{molina2021farland} offer alternatives at different levels of abstraction, and
CyGym~\citep{lanier2026cygym} models cyber defence directly as a two-player general-sum game.

Games of this size are usually solved by growing each player's strategy set gradually.
Policy-space response oracles (PSRO)~\citep{lanctot2017psro} start from a few attacker and
defender policies and repeat three steps. They estimate the payoff of every attacker--defender
pair in the current sets, which gives a small matrix game; they solve that matrix game for a
mixed equilibrium; and they train one new policy per player as a best response to the
opponent's current mixture, add it to the sets, and start again. The loop stops when neither
new policy improves on the mixture, and MetaDOAR~\citep{lanier2026metadoar} scales the
best-response step to large networks.

The three steps cost very different amounts. A best response is a bounded reinforcement
learning run. Solving the matrix game takes milliseconds with fictitious
play~\citep{brown1951fictitious} or a standard solver~\citep{savani2025gambit}. Estimating
one payoff, however, means running the full simulator many times under both policies, and
the matrix has $n_D n_A$ entries for pools of $n_D$ defender and $n_A$ attacker policies, so
the simulation bill grows quadratically over a run. Simulation throughput is already a known
bottleneck: CybORG++~\citep{emerson2024cyborgpp} was built largely for speed and reports up to
a thousandfold faster parallel execution. Faster simulators usually give up fidelity by
abstracting host, process and packet detail~\citep{standen2021cyborg,molina2021farland}, and
that detail is often the realism the analysis needed. We take the other route and keep the
simulator, but call it less often.

RWPS spends a fixed simulation budget on the payoffs that can change the equilibrium and
predicts the others with a learned model. The model is trained on every payoff measured
earlier in the run, so a simulation paid for in one iteration keeps helping in later ones.
A score chooses which payoffs to simulate, favouring those the equilibrium depends on and the
model is least sure of. The obvious worry is what this costs in equilibrium quality. The
standard analysis of estimated games~\citep{tuyls2020bounds} bounds the damage by twice the
largest payoff error over \emph{all} entries. That bound is uninformative here, because RWPS
is deliberately inaccurate on entries it judges irrelevant, and those are exactly the entries
that set the largest error. Judging the method requires a bound that separates errors the
equilibrium can feel from errors it cannot. Our contributions are as follows.

\begin{enumerate}[leftmargin=1.6em,itemsep=2pt]
  \item \textbf{Bounds for estimated games that ignore irrelevant errors.} We sharpen the
        bound of \citet{tuyls2020bounds} so that a player's regret depends only on payoff
        errors weighted by the opponent's equilibrium strategy, and give a signed version
        in which only two kinds of error can hurt (Lemma~\ref{lem:supp}). Once every
        strategy has been simulated against the opponent's equilibrium strategies, a set
        of $B^\ast = n_D s_A + n_A s_D - s_D s_A$ entries, errors in the predicted entries
        cannot affect regret at all (Corollary~\ref{cor:coverage}). We also give a
        certificate computable from simulated payoffs alone (Corollary~\ref{cor:cert}).
        On the estimator's own output the new bounds are four to six times tighter than
        the standard one.
  \item \textbf{Regret-Weighted Payoff Sampling.} RWPS predicts unsimulated payoffs with
        a model trained on every payoff measured so far, chooses what to simulate with a
        score built from bootstrapped equilibrium strategies, and keeps a small amount of
        random exploration, which the appendix shows is necessary. Inside PSRO it reaches
        lower exploitability than minimum-regret-first search, information-gain
        search~\citep{jordan2008mrfs} and progressive sampling at the same budget. The
        prediction model accounts for $62\%$ of the gain over random sampling and the
        score for the other $38\%$, and with the model held fixed the score beats the
        sampling rule of \citet{sokota2019deviation} on Colonel Blotto.
  \item \textbf{The Agentic Network Security Game.} ANSG extends
        CyGym~\citep{lanier2026cygym} to networks of LLM agents that an attacker can
        subvert through indirect prompt injection, with general-sum payoffs and partial
        observability on both sides. On CyGym and ANSG, RWPS reaches the lowest
        exploitability at the smallest budgets; on ANSG, without its score it needs four
        simulated payoffs to reach the level it otherwise reaches with one.
\end{enumerate}
\section{Related Work}
\label{sec:related}

Several simulators target cyber defence. CybORG~\citep{standen2021cyborg,cage2023} drives the
CAGE challenges. Yawning Titan~\citep{andrew2022yawningtitan},
CyberBattleSim~\citep{cyberbattlesim2021} and FARLAND~\citep{molina2021farland} model networks
at different levels of fidelity, and CyGym~\citep{lanier2026cygym} adds the game-theoretic
layer this paper builds on. In all of them, evaluating a strategy pair means running the
simulator, and that cost is a recognised obstacle. CybORG++~\citep{emerson2024cyborgpp}
reports up to a thousandfold parallel speedup from a lightweight re-implementation, built
because large-scale evaluation on the original was infeasible. Our work is complementary to
simulator engineering. We reduce the number of simulator calls the equilibrium computation
needs, and the two approaches can be combined.

The double-oracle method~\citep{mcmahan2003double} and PSRO~\citep{lanctot2017psro} grow
strategy pools by best response to the current meta-game equilibrium. Later work improves
oracle quality and diversity~\citep{balduzzi2019openended}, and studies of the structure of
these pools~\citep{czarnecki2020spinningtops} describe what equilibrium supports look like in
practice, which matters here because support size sets the cost in
Corollary~\ref{cor:coverage}. MetaDOAR~\citep{lanier2026metadoar} scales the best-response
step to large cyber networks. Little of this work addresses the evaluation step, where most
of the cost lies in simulation-based games.

Our theory refines a known result. \citet{tuyls2018generalised,tuyls2020bounds} prove that
a Nash equilibrium of an estimated meta-game is a $2\tau$-approximate equilibrium of the true
game, where $\tau$ is the largest payoff error. Lemma~\ref{lem:supp} replaces that largest
error with one weighted by the opponent's equilibrium strategy, so the bound stays
informative for an estimator that leaves irrelevant cells inaccurate on purpose.
\citet{rowland2019multiagent} allocate samples for evaluation under incomplete information,
mainly for $\alpha$-rank~\citep{omidshafiei2019alpharank} response graphs, using bandit rules
without a surrogate that generalises across cells.

Query-efficient EGTA methods reduce the same cost by pruning profiles instead of predicting
them. Progressive-sampling algorithms sample profiles in rounds and stop sampling a profile
once its confidence interval shows it cannot affect the learning
goal~\citep{areyanviqueira2019learning,areyanviqueira2020improved,cousins2022computational,mishra2023regretpruning}.
Our method instead solves a matrix in which every cell is present, either simulated or
predicted. The two guarantees promise different things. A uniform guarantee covers the whole
estimated matrix, so any equilibrium computed from it is approximately correct, but it exists
only after every surviving profile has been sampled to accuracy $\eps$. Our certificate
covers only the equilibrium actually returned and says nothing about cells we did not
simulate, but it can be computed from whatever has been simulated so far. At the budgets we
target, progressive sampling cannot prune a single profile and gives no guarantee
(Section~\ref{sec:psro-headline}), while our certificate is available at every budget and
tightens as more cells are simulated.

\citet{wellman2025egtasurvey} still describe MRFS-style heuristic search as the practical
option when exhaustive evaluation is infeasible, so we treat it as a live baseline. We
compare against both algorithms of \citet{jordan2008mrfs} and reach lower exploitability than
both in a growing pool (Section~\ref{sec:psro-headline}). Neither carries information between
iterations beyond its cache, while our surrogate trains on every cell bought in every earlier
iteration, so our advantage grows over a run. On a fixed pool, where there is no such
history, MRFS is stronger at low budgets.

On the model-learning side, \citet{wang2023mfgegta} pair double oracle with a learned game
model and regularisation for sample efficiency, although for mean-field games and without an
explicit payoff matrix, and \citet{wang2022strategyexploration} study which strategies to
add to a pool, a question separate from which profiles to evaluate. Earlier work learns
payoff or deviation-payoff functions over strategy
spaces~\citep{vorobeychik2007payoff,sokota2019deviation} within
EGTA~\citep{wellman2006egta,vorobeychik2008stochastic}.

\citet{sokota2019deviation} pair a learned deviation-payoff model with sampling around
candidate equilibria, which is the idea closest to ours, but they build no payoff table and
give no bound. \citet{picheny2016bayesian} apply Bayesian optimisation to Nash computation in
continuous black-box games with a Gaussian-process surrogate and a fixed game. Adaptive
payoff queries for $\alpha$-rank~\citep{rowland2019multiagent,rashid2021alpharank} target a
different solution concept. Active learning~\citep{settles2009active,shahriari2016bayesopt}
and deep ensembles~\citep{lakshminarayanan2017ensembles} provide our acquisition and
uncertainty tools. The appendix compares each of these in detail.

\section{Method}
\label{sec:methods}

\subsection{Preliminaries and problem statement}

A simulation-based security game is a two-player general-sum game between a defender and an
attacker. Each player holds a set of pure strategies. In our setting a pure strategy is a
complete policy for the underlying stochastic game, for instance a trained neural
controller for the defender's network actions or the attacker's intrusion actions. The game
is \emph{simulation-based} because its payoffs have no analytical form. The expected
utility of a strategy pair $(d, a)$ is
\[
U^D(d,a) = \mathbb{E}\bigl[u(d,a)\bigr],
\qquad
U^A(d,a) \text{ symmetrically},
\]
where the expectation is over the simulator's stochasticity and each sample of it costs a
full rollout. The game is general-sum, $U^D + U^A \not\equiv c$, because the two players
value the same outcome differently. Nothing below assumes zero-sum structure.

Restricted to finite pools $S_D = \{d_1,\dots,d_{n_D}\}$ and
$S_A = \{a_1,\dots,a_{n_A}\}$, the game is a bimatrix game with payoff matrices
$D_{ij} = U^D(d_i, a_j)$ and $A_{ji} = U^A(d_i, a_j)$, scaled without loss of generality to
$[0,1]$. A solution is a mixed-strategy Nash equilibrium, a pair of distributions
$(p, q) \in \Delta_{n_D} \times \Delta_{n_A}$ from which neither player can profitably
deviate,
\[
p^\top D q \ge e_i^\top D q \;\;\forall i,
\qquad
q^\top A p \ge e_j^\top A p \;\;\forall j .
\]
Quality is measured by exploitability, the largest deviation gain either player forgoes,
\begin{equation}
\eps(p,q) = \max\Bigl\{\max_i\, (Dq)_i - p^\top D q,\;\;
                       \max_j\, (Ap)_j - q^\top A p\Bigr\},
\label{eq:eps-def-m}
\end{equation}
so $\eps = 0$ is an exact equilibrium and, with scaled payoffs, $\eps$ reads as a fraction
of the payoff range. A profile whose exploitability is at most 
$\eps$ is called an $\eps$-Nash equilibrium.

PSRO solves the full game by repeating three operations. It solves the current restricted
game for $(p, q)$, trains a best response to each opponent mixture with a
reinforcement-learning oracle and adds it to the pool, and then extends the payoff matrices
with the new strategy's entries. In the simulation-based setting the first operation is
nearly free, the second is a bounded training run, and the third costs
$N_{\mathrm{MC}}$ rollouts per cell over $\Theta(n_D n_A)$ cells. This paper addresses
the third step under a budget. Given pools $(S_D, S_A)$, a cache of previously simulated
entries, and a budget of $B \ll n_D n_A$ new cell evaluations, the goal is to produce estimates $(\Dhat, \Ahat)$ whose equilibrium
is close, in the sense of Eq.~\eqref{eq:eps-def-m}, to an equilibrium of the true
$(D, A)$.

The design question is which $B$ cells to simulate and what to put in the others.

\begin{figure}[t]
\centering
\begin{tikzpicture}[
  scale=0.84, transform shape,
  font=\tiny,
  meas/.style={draw=black!45, fill=black!12, minimum size=6.0mm, inner sep=0pt},
  pred/.style={draw=black!45, fill=white, minimum size=6.0mm, inner sep=0pt,
               text=black!55, font=\tiny\itshape},
  pick/.style={draw=blue!70, fill=blue!22, line width=0.5pt, minimum size=6.0mm,
               inner sep=0pt, text=black!75, font=\tiny\itshape},
  box/.style={draw=black!55, rounded corners=0.7mm, align=center, inner sep=1.0mm},
  bar/.style={fill=black!22, draw=none},
  topbar/.style={fill=blue!45, draw=none},
  flow/.style={-{Latex[length=1.5mm]}, draw=black!70, line width=0.5pt},
  lab/.style={font=\tiny, inner sep=0.4mm},
]

\node[lab, anchor=west] at (-2mm, 18.5mm) {\textbf{PSRO\,/\,MetaDOAR iteration $t$}};
\draw[black!25] (-2mm,16.9mm) -- (63mm,16.9mm);

\node[box, anchor=north west] (emb) at (-2mm, 14.2mm)
     {$\varphi(d_i)=[\,\cdot\,\cdot\,\cdot\,]$\\[-0.3mm]
      $\varphi(a_j)=[\,\cdot\,\cdot\,\cdot\,]$};
\node[lab, below=0.3mm of emb] {strategy embeddings};

\node[box, anchor=west, minimum height=7.5mm] (sur) at (19mm, 11.0mm)
     {ensemble\\ surrogate};
\node[lab, below=0.3mm of sur] {$\hat u_{ij},\ \hat\sigma_{ij}$};

\begin{scope}[xshift=40mm, yshift=13.0mm]
  \node[meas] (m11) at (0,0)      {$.62$};
  \node[meas] (m12) at (6.2mm,0)  {$.31$};
  \node[pred] (m13) at (12.4mm,0) {$.45$};
  \node[meas] (m21) at (0,-6.2mm)      {$.18$};
  \node[pred] (m22) at (6.2mm,-6.2mm)  {$.52$};
  \node[pick] (m23) at (12.4mm,-6.2mm) {$.37$};
  \node[pred] (m31) at (0,-12.4mm)      {$.29$};
  \node[pred] (m32) at (6.2mm,-12.4mm)  {$.44$};
  \node[pred] (m33) at (12.4mm,-12.4mm) {$.21$};
  \node[lab, above=0.4mm of m12] {payoff matrix $\Dhat$};
\end{scope}

\begin{scope}[xshift=-2mm, yshift=-10.5mm]
  \foreach \k/\o in {0/0, 1/0.9, 2/1.8}{
    \draw[draw=black!40, fill=white] (\o mm, -\o mm) rectangle ++(7mm,7mm);
  }
  \draw[black!30] (1.8mm,-1.8mm) ++(0,2.34mm) -- ++(7mm,0);
  \draw[black!30] (1.8mm,-1.8mm) ++(0,4.66mm) -- ++(7mm,0);
  \draw[black!30] (1.8mm,-1.8mm) ++(2.34mm,0) -- ++(0,7mm);
  \draw[black!30] (1.8mm,-1.8mm) ++(4.66mm,0) -- ++(0,7mm);
  \node[lab, anchor=north, align=center] at (4.4mm,-3.4mm)
       {bootstrap\\ re-solves};
\end{scope}

\node[box, anchor=west] (marg) at (13.5mm,-8.6mm) {$\tilde p,\ \tilde q$};
\node[lab, below=0.3mm of marg, align=center] {marginals,\\ $\nu$-smoothed};

\node[lab, anchor=south west] at (23mm,-6.4mm) {scores $\tilde p_i\tilde q_j\hat\sigma_{ij}$};
\draw[topbar] (23mm,-9.6mm) rectangle (32mm,-7.4mm);
\draw[bar]    (23mm,-12.4mm) rectangle (28.4mm,-10.2mm);
\draw[bar]    (23mm,-15.2mm) rectangle (25.5mm,-13.0mm);

\begin{scope}[xshift=52mm, yshift=-10.0mm]
  \node[draw=black!55, rounded corners=0.7mm, minimum width=11mm,
        minimum height=9.5mm] (sim) at (0,0) {};
  \fill[black!65] (-2.9mm,2.2mm) circle (0.6mm);
  \fill[black!65] ( 2.9mm,2.2mm) circle (0.6mm);
  \fill[black!65] (-2.9mm,-2.2mm) circle (0.6mm);
  \fill[black!65] ( 2.9mm,-2.2mm) circle (0.6mm);
  \fill[black!65] (0,0) circle (0.6mm);
  \draw[black!45] (-2.9mm,2.2mm) -- (0,0) -- (2.9mm,2.2mm);
  \draw[black!45] (-2.9mm,-2.2mm) -- (0,0) -- (2.9mm,-2.2mm);
  \draw[black!45] (-2.9mm,2.2mm) -- (-2.9mm,-2.2mm);
  \node[lab, below=0.4mm of sim] {simulator};
\end{scope}

\draw[flow] (emb.east) -- (emb.east -| sur.west);
\draw[flow] (sur.east) -- (36.6mm,11.0mm) node[lab, midway, above] {fill};
\draw[flow] (36.6mm,0.6mm) -| (3.5mm,-3.0mm)
      node[lab, pos=0.32, above] {resample and re-solve};
\draw[flow] (7.5mm,-7.2mm) -- (13.1mm,-8.6mm);
\draw[flow] (19.9mm,-8.6mm) -- (22.6mm,-8.6mm);
\draw[flow] (32.4mm,-9.0mm) -- (46.0mm,-9.0mm)
      node[lab, midway, above] {buy, $p = 1-\eps_{\mathrm{x}}$};
\draw[flow, densely dashed] (28.0mm,-13.3mm) -- (46.0mm,-12.0mm)
      node[lab, pos=0.68, below] {random, $p = \eps_{\mathrm{x}}$};
\draw[flow] (58.0mm,-10.0mm) -- (61.5mm,-10.0mm);
\draw[flow] (61.5mm,-10.0mm) -- (61.5mm,6.8mm) -- (55.6mm,6.8mm);
\node[lab, anchor=west, align=left] at (62.0mm,-1.0mm) {measured\\ value};
\draw[flow] (56.5mm,13.0mm) .. controls +(4mm,0) and +(0,-4mm) .. (58mm,18.9mm);
\node[lab, anchor=west, align=left] at (58.4mm,18.9mm)
     {cache and surrogate\\ to iteration $t{+}1$};
\end{tikzpicture}
\caption{One acquisition round at PSRO\,/\,MetaDOAR iteration $t$. Embeddings feed an ensemble surrogate
that fills every uncached entry of $\Dhat$ (italic) beside the measured ones (grey).
Re-solving bootstrap realizations of the filled matrix gives equilibrium marginals, mixed
toward uniform at rate $\nu$; each uncached cell is scored by its marginal weight times the
surrogate's uncertainty, the
top-scoring cell is simulated, and its value joins the cache carried into iteration
$t{+}1$.}
\label{fig:overview}
\end{figure}
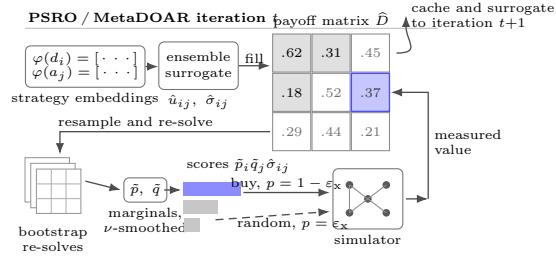

\subsection{The budgeted estimator}
\label{sec:estimator}

Algorithm~\ref{alg:build} gives the full procedure. It has four steps, explained below in
order: embed each strategy, fit a surrogate, score cells by their influence on the
equilibrium, and simulate the top-scoring cells with some forced exploration, which the
appendix shows is necessary.\label{sec:embedding}\label{sec:surrogate}\label{sec:acquisition}\label{sec:exploration}
We call this \emph{Regret-Weighted Payoff Sampling} (RWPS), because the budget goes to
cells in proportion to their estimated influence on equilibrium regret, measured through
smoothed bootstrap marginals. The method has two separable parts, the history-trained fill
and the acquisition score. Section~\ref{sec:acq-ablation} measures each. The fill accounts
for most of the gain, and the score matters most at the smallest budgets.

\begin{algorithm}[t]
\caption{Regret-Weighted Payoff Sampling (RWPS)}
\label{alg:build}
\begin{algorithmic}[1]
\Require pools $S_D, S_A$; cache $\mathcal{C}$ of simulated entries; budget $B$; batch
size $c$; exploration probability $\eps_{\mathrm{x}}$; smoothing $\nu$; bootstrap size
$B_{\mathrm{boot}}$; ensemble size $M$
\State $\varphi(s) \gets$ embedding of every $s \in S_D \cup S_A$
       \Comment{no rollouts}
\State $B \gets B - |\{(i,j) \in \mathcal{C}\}|$ if $B$ is a coverage target
       \Comment{spend only on the \emph{uncovered remainder}}
\For{$r = 1, \dots, \lceil B/c \rceil$}
  \State fit ensemble $\{f_\psi^{(m)}\}_{m=1}^{M}$ on
         $\{(\varphi(d_i), \varphi(a_j)) \mapsto \mathcal{C}[i,j]\}$
  \State $(\hat u_{ij}, \hat\sigma_{ij}) \gets$ ensemble mean, spread for all
         $(i,j) \notin \mathcal{C}$
  \For{$b = 1, \dots, B_{\mathrm{boot}}$}
     \Comment{uncertainty of the \emph{solution}}
     \State subsample $k_D \times k_A$ strategies; draw
            $\widetilde D_{ij} \sim \mathcal{N}(\hat u_{ij}, \hat\sigma_{ij}^2)$ on
            unsimulated cells, pin cached cells
     \State $(p^{(b)}, q^{(b)}, v^{(b)}) \gets$ Nash solver on the realization
  \EndFor
  \State $\bar p_i \gets \operatorname{mean}_b\, p^{(b)}_i$;\quad
         $\bar q_j \gets \operatorname{mean}_b\, q^{(b)}_j$
         \Comment{product of marginals}
  \State $\tilde p \gets (1-\nu)\,\bar p + \nu/|S_D|$;\quad
         $\tilde q \gets (1-\nu)\,\bar q + \nu/|S_A|$
         \Comment{smooth toward uniform}
  \State $\alpha(i,j) \gets \tilde p_i\, \tilde q_j\, \hat\sigma_{ij}$
         \Comment{zero on cached cells}
  \For{$t = 1, \dots, c$}
     \Comment{$c$ cells per refit; the surrogate is refit between batches}
     \State with probability $\eps_{\mathrm{x}}$: simulate a cell drawn uniformly at
            random, unsimulated first, then revisits averaged by running mean
            \Comment{necessary; see appendix}
     \State otherwise: simulate the unsimulated cell of largest $\alpha$
     \State add to $\mathcal{C}$
  \EndFor
\EndFor
\State $\Dhat_{ij}, \Ahat_{ji} \gets \mathcal{C}[i,j]$ if cached, else ensemble mean
\State \Return $(\Dhat, \Ahat)$
\end{algorithmic}
\end{algorithm}

First, each strategy receives an embedding $\varphi(s) \in \mathbb{R}^d$ computed
without simulation. RWPS requires only that $\varphi$ distinguish the strategies in the pool. For policies with an actor $\mu_\pi$, a natural choice is a
behavioural fingerprint over a fixed probe bank $\Xi = \{\xi_1,\dots,\xi_P\}$ of
states,
\begin{equation}
\varphi(\pi) = \bigl[\operatorname{mean}_p \mu_\pi(\xi_p)\,\Vert\,
                     \operatorname{std}_p \mu_\pi(\xi_p)\bigr],
\label{eq:fingerprint}
\end{equation}
at the cost of $P$ forward passes. The simplest choice is an identity tag, a one-hot
vector per strategy, which carries no behavioural information.
Section~\ref{sec:experiments} states the embedding used in each experiment.

Second, an ensemble of $M$ small networks is fit to map embedding pairs to payoff pairs,
each on a bootstrap resample of every entry ever simulated, so a payoff bought at one PSRO
iteration still trains the surrogate at the next. Queries return the mean $\hat u_{ij}$ and
disagreement $\hat\sigma_{ij}$. Cached entries override surrogate values, and a rank-one
row-and-column model stands in until enough labels exist.

Third, each unsimulated cell is scored by how much the equilibrium uses it, times how
uncertain the surrogate is about it. We weight by equilibrium use because only those
cells can change the answer. By Lemma~\ref{lem:supp}, a player's regret depends on
payoff errors only through the opponent's equilibrium mixture, and with the support
held fixed a change to cell $(i,j)$ moves the defender's value $p^\top D q$ by $p_i q_j$
times that change. A cell that no equilibrium strategy uses cannot affect the solution however uncertain it
is, so scoring by uncertainty alone wastes budget. The weights depend on the unknown
equilibrium, but re-solving is nearly free, so we estimate it by re-solving
$B_{\mathrm{boot}}$ bootstrap realizations of the game, each with a resampled pool and a
matrix realization drawn from the surrogate's predictive distribution.

Weighting by $p_i q_j$ from a single solve would score only cells whose strategies that
solve already plays, so an error in the estimated support would never be corrected. We
make two choices so that cells outside the estimated support can also be scored. First,
we average each player's mixture across the bootstrap replicates and then multiply,
$\bar p_i \bar q_j$, instead of averaging the product within each replicate. If one
replicate plays $d_1$ against $a_1$ and another plays $d_2$ against $a_2$, averaging the
products scores only those two cells, while multiplying the averages also scores the
cross pairs $(d_1, a_2)$ and $(d_2, a_1)$. Those are the cells that matter when the
replicates disagree about the support. Second, each averaged mixture is blended with
uniform, $\tilde p = (1-\nu)\bar p + \nu/|S_D|$, with a hyperparameter $\nu$, so every
strategy keeps some weight even when all replicates agree and are wrong. Setting
$\nu = 0$ and averaging the products instead recovers value-sensitivity weighting,
reported as an ablation in the appendix.

Finally, cells are simulated one selection at a time. With probability
$1-\eps_{\mathrm{x}}$ the highest-scoring unsimulated cell is bought, and with probability
$\eps_{\mathrm{x}}$ a cell chosen uniformly at random is bought instead. The random
purchases are needed because the surrogate cannot know where it is wrong on cells it has
never simulated. Two strategies can share an embedding while differing arbitrarily in
payoff, so any rule that consults only the embedding and the cache can miss the cell
that matters, as the appendix proves.

The cache persists across PSRO iterations. Cached values override the surrogate, every
fit trains on all cells bought earlier, and a coverage-target budget pays only for the
part of the deviation-relevant set not already held.
Sokota et al.'s neighbourhood sampler moves mass toward the simplex edges, whereas
$\nu$-smoothing moves it toward uniform. The appendix gives further
implementation details, such as when the surrogate is refit, and all hyperparameters.
\section{Theory: an Instance-Dependent Bound}
\label{sec:theory}

We now ask what an equilibrium of the estimated game $(\Dhat,\Ahat)$ is worth in the true
game $(D,A)$. The known answer is the sup-norm bound of
\citet{tuyls2018generalised,tuyls2020bounds}. Writing $\Delta_D = \Dhat - D$,
$\Delta_A = \Ahat - A$ and $\tau = \max\{\|\Delta_D\|_\infty, \|\Delta_A\|_\infty\}$, any
exact equilibrium of the estimated game is a $2\tau$-equilibrium of the true one. That bound is too loose for a budgeted estimator. Our estimator leaves irrelevant cells
inaccurate on purpose, so $\tau$ is set by the cells that do not matter, and the bound is
correspondingly weak (Section~\ref{sec:results}). In this section we let the equilibrium
determine which errors count, strengthening the bound step by step from a two-sided worst
case to a signed bound that can be computed from the data the estimator already holds.
Proofs of the two results below are given here; the remaining proofs, and fuller versions
of these, are in the appendix.

\subsection{Support-weighted perturbation}

\begin{lemma}[Support-weighted perturbation]
\label{lem:supp}
Let $(\hat p, \hat q)$ be an $\eps_{\mathrm{solve}}$-approximate Nash
equilibrium of $(\Dhat, \Ahat)$. Write
$\Delta_D = \Dhat - D$ and define the opponent-support-weighted errors
$\tau_D(\hat q) = \max_i(|\Delta_D|\hat q)_i$ and
$\tau_A(\hat p) = \max_j(|\Delta_A|\hat p)_j$, with $|\cdot|$ entrywise. Then
the defender's regret in the true game satisfies
\begin{equation}
\label{eq:chain}
\begin{aligned}
\mathrm{reg}_D - \eps_{\mathrm{solve}}
&\;\le\;
\underbrace{\max_i\bigl(-\Delta_D\hat q\bigr)_i
            + \hat p^\top \Delta_D \hat q}_{\text{(a) signed}} \\[4pt]
&\;\le\;
\underbrace{\tau_D(\hat q)
            + \hat p^\top|\Delta_D|\hat q}_{\text{(b) two-term}} \\[4pt]
&\;\le\;
\underbrace{2\,\tau_D(\hat q)}_{\text{(c) support-weighted}}
\;\le\;
\underbrace{2\,\tau}_{\text{(d) sup-norm}},
\end{aligned}
\end{equation}
where $\tau = \max\{\|\Delta_D\|_\infty,\|\Delta_A\|_\infty\}$, and
symmetrically for the attacker. Consequently $(\hat p,\hat q)$ is an
$\eps$-Nash equilibrium of $(D,A)$ for $\eps$ equal to
$\eps_{\mathrm{solve}}$ plus the maximum over players of any of
(a)--(d). Equality in the last step holds only when a worst-error cell lies in
a played column, respectively row.
\end{lemma}

\begin{proof}
For any pure deviation $i$ of the defender, $e_i^\top D \hat q
= e_i^\top \Dhat \hat q - e_i^\top \Delta_D \hat q
\le \hat p^\top \Dhat \hat q + (|\Delta_D|\hat q)_i
= \hat p^\top D \hat q + \hat p^\top \Delta_D \hat q + (|\Delta_D|\hat q)_i$,
using optimality of $\hat p$ in the estimated game and the triangle inequality. Taking the
maximum over $i$ and bounding $\hat p^\top \Delta_D \hat q \le \tau_D(\hat q)$ gives the
defender's bound; the attacker's is symmetric, and no zero-sum structure is used. Each row of
$|\Delta_D|\hat q$ is a $\hat q$-weighted average of entries bounded by $\tau$, so
$\tau_D(\hat q) \le \tau$: columns outside $\mathrm{supp}(\hat q)$ contribute nothing however
large their error. The signed form follows by dropping absolute values and keeping only the
two harmful directions. Appendix~A.1 gives the step-by-step version.
\end{proof}

The chain contains four bounds.\footnote{The bound constrains $\Delta_D$, not its
origin, so it also covers irrational players and model misspecification. We do not consider
either here. Throughout, $\Delta_D$ is sampling and surrogate error.} Form (d) is the sup-norm bound of
\citet{tuyls2018generalised,tuyls2020bounds}, recovered as the loosest link. Form (c) is its
support-weighted refinement, and form (b) keeps both terms instead of doubling the larger.
Form (a) drops the absolute values, so that only \emph{underestimated deviations} and
\emph{overestimated support} contribute.

We use different parts of the chain at different points below, and the parts are not
interchangeable. Theorem~\ref{thm:eps} and Corollary~\ref{cor:coverage} are derived from
(b), since the absolute values are what allow a row to be split over simulated and surrogate
cells with each class bounded by a magnitude. Corollary~\ref{cor:cert} instead rests on the
sign structure of (a), which allows an optimistic fill. The solver term
$\eps_{\mathrm{solve}}$ is computable without ground truth, since it is the regret of
$(\hat p,\hat q)$ in the matrices the estimator already holds. We report it
alongside the bounds in Section~\ref{sec:results}. It matters most for form (a), which has
no slack to absorb it.

\subsection{Finite-budget guarantee}

Under the estimator, entries are of two kinds. Simulated cells carry statistical error, and
surrogate cells carry representational error. The support-weighted bound splits accordingly.
Assume rollouts return truth plus independent $\sigma$-sub-Gaussian noise, each simulated
cell holds an average of at least $m$ visits, and surrogate cells have absolute error at
most $\beta$ (from regression error, an uninformative embedding, or both; we do not certify $\beta$ in
advance and instead measure it in Section~\ref{sec:results}). Which cells are simulated,
and how often, is chosen by the acquisition rule from earlier samples, so the design is
\emph{adaptive}. A fixed-design union bound would be optimistic under
adaptive selection. We therefore state the statistical term with a time-uniform
concentration bound~\citep{howard2021timeuniform}, valid simultaneously for all cells and
all stopping times, at the cost of a $\log\log$ factor in the visit count. The guarantee
below then holds however the acquisition rule selects cells.

\begin{theorem}[Finite-budget guarantee]
\label{thm:eps}
With probability at least $1-\delta$, simultaneously over all cells and all sample-dependent
stopping rules, any $\eps_{\mathrm{solve}}$-approximate equilibrium $(\hat p, \hat q)$ of
$(\Dhat, \Ahat)$ satisfies
\[
\eps \;\le\; \eps_{\mathrm{solve}} + 2\max_{\pi \in \{\hat p, \hat q\}} \Bigl[
      \zeta(m,\delta)
      \;+\; \beta \cdot w_{\mathrm{sur}}(\pi) \Bigr],
\]
where
$\zeta(m,\delta) = \sigma\sqrt{\dfrac{2\bigl(\log(4 n_D n_A/\delta) + \log\log_2(2m)\bigr)}{m}}$
is a time-uniform confidence radius valid under adaptive sampling, and
$w_{\mathrm{sur}}(\hat q) = \max_i \sum_{j \,:\, (i,j)\ \mathrm{surrogate}} \hat q_j$
is the largest opponent-support mass any deviation row places on surrogate cells, and
symmetrically for $w_{\mathrm{sur}}(\hat p)$.
\end{theorem}

\begin{proof}
Split $(|\Delta_D|\hat q)_i$ over the simulated and surrogate cells of row $i$. The simulated
part is at most the concentration radius times $\sum_j \hat q_j \le 1$, by concentration of
$m$-visit averages and a union bound over all $2n_Dn_A$ cells; the surrogate part is at most
$\beta$ times the $\hat q$-mass that row $i$ places on surrogate cells, which is
$w_{\mathrm{sur}}(\hat q)$. Take the maximum over rows, apply Lemma~\ref{lem:supp}, and
repeat for the attacker; Appendix~A.1 expands each step.
\end{proof}

The theorem tells acquisition what to aim for, which is to drive $w_{\mathrm{sur}}$ to zero.
Representational error enters in proportion to how much equilibrium mass touches surrogate
cells, where the sup-norm bound charges it at full strength. A surrogate cell in a column the opponent never plays costs nothing,
while one in a played column costs in proportion to that column's weight. RWPS therefore scores cells by equilibrium sensitivity as well as uncertainty. The weight
$\tilde q_j$ is zero on the columns that cannot make $w_{\mathrm{sur}}$ positive, and
largest on those that dominate it. Computing $w_{\mathrm{sur}}$ needs only a cache query
against the returned mixture, so we can check directly whether it has reached zero, and Appendix~A.8 builds a terminating coverage iteration on it. Once it vanishes the
bound keeps only its statistical term, which Appendix~A.3.2 inverts into the budget a target
exploitability requires.

\subsection{An optimistic certificate}

Only two events produce regret. The estimate can \emph{under}state a deviation row, or it
can \emph{over}state the value of the support. Filling unsimulated cells of a player's own matrix
at the maximum payoff therefore understates no deviation, which makes the following bound
computable without the true matrices. Its proof, like those of the other results in this
section, is in Appendix~A.1.

\begin{corollary}[Computable optimistic certificate]
\label{cor:cert}
Let $(\hat p, \hat q)$ be an $\eps_{\mathrm{solve}}$-approximate equilibrium of
estimated matrices $(\Dhat, \Ahat)$ whose support block
$\mathrm{supp}(\hat p)\times\mathrm{supp}(\hat q)$ is simulated, and let
$(D^{\mathrm{opt}}, A^{\mathrm{opt}})$ equal the measured values on simulated cells and the
maximum payoff elsewhere. If rollout noise is $\sigma$-sub-Gaussian and each simulated cell
averages $m$ visits, then with probability at least $1-\delta$,
\begin{align*}
\eps \;\le\; \max\Bigl\{\,&\max_i \bigl(D^{\mathrm{opt}}\hat q\bigr)_i
              - \hat p^\top \Dhat \hat q,\;\; \\
              &\max_j \bigl(A^{\mathrm{opt}}\hat p\bigr)_j
              - \hat q^\top \Ahat \hat p \Bigr\}
      \;+\; 2\,\zeta_{m,\delta} + \eps_{\mathrm{solve}},
\end{align*}
where $\zeta_{m,\delta} = \zeta(m,\delta)$ is the same time-uniform radius as in
Theorem~\ref{thm:eps}, valid under adaptive sampling.
\end{corollary}

Every quantity on the right is available at the end of a build. The deviation terms use
$(D^{\mathrm{opt}}, A^{\mathrm{opt}})$, the value terms use the estimated matrices
$(\Dhat, \Ahat)$ the solver was given, and $\eps_{\mathrm{solve}}$ is measured on those, so
the true matrices never appear. The surrogate error $\beta$ of Theorem~\ref{thm:eps} is the
one quantity no data can bound, because it concerns cells never simulated. Covering the
deviation-relevant set sets $w_{\mathrm{sur}} = 0$ and removes it.

\begin{corollary}[Deviation-relevant coverage]
\label{cor:coverage}
Let $s_D = |\mathrm{supp}(\hat p)|$, $s_A = |\mathrm{supp}(\hat q)|$. If every cell of
$\{1{:}n_D\} \times \mathrm{supp}(\hat q)$ and of
$\mathrm{supp}(\hat p) \times \{1{:}n_A\}$ is simulated, then
$w_{\mathrm{sur}}(\hat p) = w_{\mathrm{sur}}(\hat q) = 0$ and the bound of
Theorem~\ref{thm:eps} reduces to its statistical and solver terms, independent
of surrogate quality.
The set has
\[
B^\ast \;=\; n_D\, s_A + n_A\, s_D - s_D s_A
\]
cells, against $n_D n_A$ for a full rebuild.
\end{corollary}

The principle that a profile's regret depends only on its unilateral deviations
is standard in EGTA~\citep{wellman2006egta,wellman2025egtasurvey} and underlies
MRFS~\citep{jordan2008mrfs}. Setting $s_D=s_A=1$ recovers its $2n-1$ confirmation set. We
add the closed-form bimatrix count and the link to Theorem~\ref{thm:eps}. In earlier
methods the cells outside the deviation set are never evaluated. Here they are
\emph{filled} by a surrogate the solver reads, and the
corollary shows that the fill cannot reach either player's regret, which is why predicting
unsimulated payoffs is safe where earlier work pruned them. Because the supports it needs
are estimated by re-solving bootstrap samples of the game, without any rollouts, $B^\ast$ also works as a
screen before any simulation. A small $B^\ast$ says the estimator is worth enabling, and a
$B^\ast$ close to $n_D n_A$ says a full rebuild is the better option. The appendix prices a target
certificate from $B^\ast$ and gives the extended discussion.

Exploration is necessary for the asymptotic guarantee. An embedding is a finite summary,
so two policies can share one while differing arbitrarily in payoff, and any deterministic
acquisition rule that reads only embeddings and the cache can be led never to simulate the
cell that matters, as the appendix proves. The uniform draws of Algorithm~\ref{alg:build} are
the only mechanism that does not consult the embedding.

\section{Experimental Setup}
\label{sec:experiments}

The experiments ask whether the estimator finds better equilibria than other budgeted
payoff builds inside growing-pool PSRO (Section~\ref{sec:psro-headline}), whether that
holds on cyber simulators (Section~\ref{sec:cyber-sweep}), and whether the bounds of
Section~\ref{sec:theory} are tight enough to be worth computing
(Section~\ref{sec:tightness}). On the latent-quality games and Blotto,
exploitability is measured against the true, computable Nash equilibrium. On CyGym and
ANSG no such equilibrium is available, which is why RWPS is needed there in the first
place, so we measure exploitability against a full rebuild of a fixed pool instead.

\subsection{Games}
All are $21 \times 21$ ($441$ cells), payoffs scaled to $[0,1]$, rollout noise
$\sigma = 0.10$ of the range, $N_{\mathrm{MC}} = 4$ episodes per cell evaluation, so a full
rebuild costs $1764$ episodes. They differ in what the fingerprint encodes and in
equilibrium support size, the two quantities the theory says matter.

\textbf{Latent quality.} Strategies carry latent qualities $q_s \sim \mathcal{N}(0,1)$
with bilinear general-sum payoffs $D_{ij} = q^D_i q^A_j + 0.3 q^D_i$,
$A_{ij} = -q^D_i q^A_j + 0.3 q^A_j$. Equilibrium supports are $(2,2)$, so
$B^\ast = 80$ cells, $18\%$. Two embeddings share these payoffs: an informative one,
$[q, \tanh q, q^2, \sin q]$, which makes payoff smooth in $\varphi$, and an arbitrary one
drawn independently of $q$, which ablates the embedding
(appendix).

\textbf{Colonel Blotto, asymmetric values.} $S = 5$ soldiers over $F = 3$ fields, fields
worth different amounts to the two players drawn from $[0.5,2]$, hence general-sum
($\operatorname{std}(D{+}A) = 0.21$ across cells). The fingerprint is the normalised
allocation; payoff is combinatorial in it. Supports are $(12,12)$, so $B^\ast = 360$ cells,
or $82\%$, and the theory predicts before any run that this game is expensive.

\textbf{CyGym and ANSG.} CyGym~\citep{lanier2026cygym} is a network security game in
which an attacker compromises nodes by gaining binary access and a defender responds
over a partially observed network. The Agentic Network Security Game (ANSG) extends it
to networks whose nodes are LLM agents. Compromise becomes \emph{goal subversion}. An
attacker injects instructions into content an agent reads, by indirect prompt injection
or malicious tool triggering, so the agent pursues attacker goals while its legitimate
task may still complete. Subverted agents can propagate payloads downstream. The defender
audits, reprompts, resets or isolates nodes on noisy signals, and the attacker cannot
observe resets and must probe to learn which nodes remain subverted. Valuations are independent
per player, so the game is general-sum. ANSG keeps CyGym's per-node state and PSRO
machinery. Full details are in the appendix. In the CyGym and ANSG budget sweeps every strategy is embedded by an identity tag.

\subsection{Measurements}
\textbf{Exploitability} $\eps$, Eq.~\eqref{eq:eps-def}, of the returned mixture in the true
game; lower is better, $0$ is exact, and scaling makes it a fraction of the payoff range:
\begin{equation}
\eps = \max\Bigl\{\max_i (Dq)_i - p^\top D q,\;\; \max_j (Ap)_j - q^\top A p\Bigr\}.
\label{eq:eps-def}
\end{equation}
\textbf{Bound tightness}: realised $\eps$ against $2\tau$, $2\tsupp$, and the sharp form of
Lemma~\ref{lem:supp}, all computed from the method's own output. Error decomposition and coverage studies are in the appendix.

\subsection{Baselines and protocol}
\emph{Uniform sampling} spends the same budget uniformly at random with a mean fill;
\emph{full rebuild} simulates every cell, as standard PSRO does. We additionally compare
against minimum-regret-first search and information-gain search~\citep{jordan2008mrfs} and
the progressive-sampling family (Section~\ref{sec:psro-headline}). The latter shows that a
uniform $\eps$-approximation guarantee cannot be bought at these budgets. Restricted games are solved by fictitious play on the
synthetic games and on CyGym and ANSG by support enumeration with a Lemke--Howson
fallback. The solver's residual regret enters every bound as $\eps_{\mathrm{solve}}$. The growing-pool
comparison uses a $5\%$ per-build budget over twelve iterations with pools growing from three
to fourteen per side, sixteen seeds; the fixed-pool diagnostic builds of
Section~\ref{sec:tightness} sweep $\{5,10,20,40,70\}\%$ of the matrix over sixteen seeds, measuring $\eps$ against the full game
by zero-lifting the restricted mixture, with best-response episodes counted separately. All
runs are deterministic given recorded seeds. The $\tau$-statistics are maxima over $441$
cells and therefore heavy-tailed, so we report medians.

\section{Results}
\label{sec:results}

All aggregates are mean $\pm$ one standard deviation over sixteen seeds unless a caption
states otherwise. Every method in a comparison starts from the same pool.

\subsection{Exploitability over PSRO Rounds for Synthetic Games}
\label{sec:psro-headline}

This is the setting the estimator is built for. Pools grow, the cache persists, and every
payoff entry costs a rollout. Table~\ref{tab:psro} reports twelve PSRO
iterations at a $5\%$ per-build budget over sixteen seeds, with exploitability
measured against the full $21\times21$ game by zero-lifting the restricted
mixture.

\begin{figure*}[t]
\centering
\includegraphics[width=0.76\textwidth]{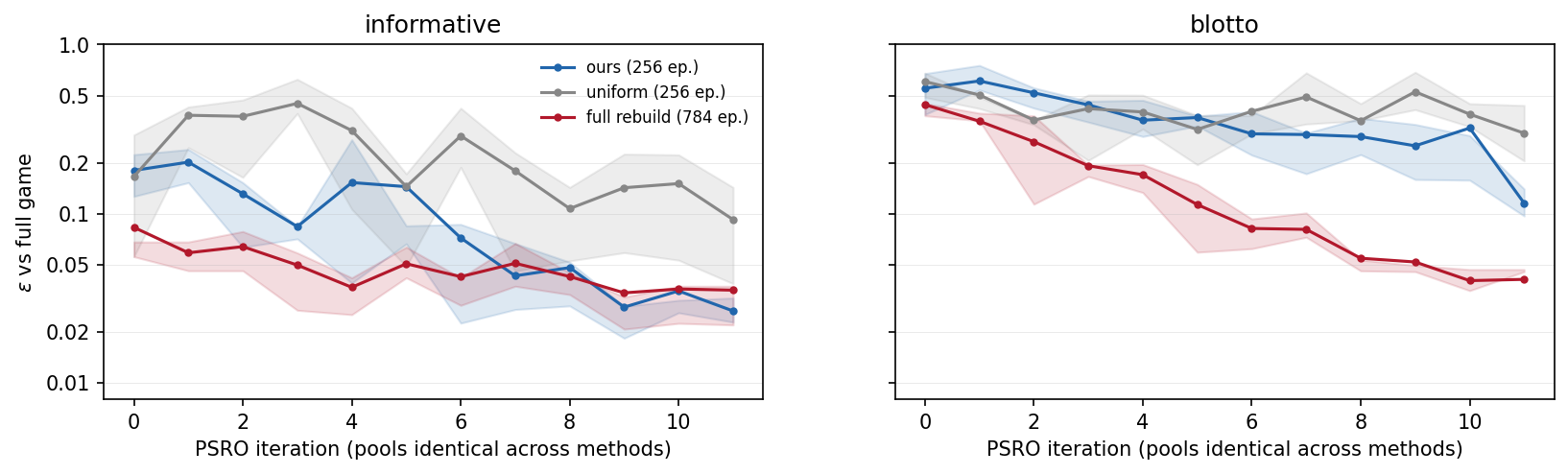}
\caption{Exploitability against the full game over twelve PSRO iterations at a $5\%$
per-build payoff budget. All methods draw the same initial pool and share a best-response
oracle, but each responds to its own mixture, so the pools coincide at the first iteration
and diverge thereafter. Means over sixteen seeds, interquartile bands.}
\label{fig:psro}
\end{figure*}

\begin{table}[t]
\centering\small
\caption{Growing-pool PSRO, $5\%$ per-build budget, mean $\pm$ SD over sixteen
seeds. $\eps_0$ is the first iteration, where every method holds the same randomly
drawn pool. Budget is cumulative payoff-rollout episodes over the run, and every
budgeted method is held to the same $256$. The full rebuild has no budget. It simulates
every uncached cell, costing $784$ episodes by iteration twelve, and is shown as a quality
reference.
The three progressive-sampling members coincide because none prunes a single
index at this budget.}
\label{tab:psro}
\begin{tabular}{@{}llccr@{}}
\toprule
Game & Payoff build & $\eps_0$ & $\eps_{12}$ & budget \\
\midrule
\multirow{6}{*}{Informative}
 & RWPS (ours)            & $0.262$ & $\mathbf{0.024 \pm 0.017}$ & $256$ \\
 & MRFS                   & $0.180$ & $0.091 {\scriptstyle\pm} 0.099$ & $256$ \\
 & IGS                    & $0.168$ & $0.113 {\scriptstyle\pm} 0.111$ & $256$ \\
 & uniform (flat fill)    & $0.208$ & $0.113 {\scriptstyle\pm} 0.097$ & $256$ \\
 & GS / PSP / PS-REG-M    & $0.223$ & $0.223 {\scriptstyle\pm} 0.174$ & $256$ \\
 & \emph{full rebuild}    & $0.111$ & $0.036 {\scriptstyle\pm} 0.023$ & $784$ \\
\midrule
\multirow{6}{*}{Blotto}
 & RWPS (ours)            & $0.525$ & $\mathbf{0.082 \pm 0.042}$ & $256$ \\
 & MRFS                   & $0.515$ & $0.222 {\scriptstyle\pm} 0.114$ & $256$ \\
 & IGS                    & $0.554$ & $0.276 {\scriptstyle\pm} 0.133$ & $256$ \\
 & uniform (flat fill)    & $0.593$ & $0.293 {\scriptstyle\pm} 0.140$ & $256$ \\
 & GS / PSP / PS-REG-M    & $0.627$ & $0.627 {\scriptstyle\pm} 0.180$ & $256$ \\
 & \emph{full rebuild}    & $0.398$ & $0.045 {\scriptstyle\pm} 0.015$ & $784$ \\
\bottomrule
\end{tabular}
\end{table}

The gap between methods builds up over the run.
Figure~\ref{fig:psro} plots the same runs iteration by iteration. All methods begin
together, separate over the first four iterations, and keep their ordering after that.

By iteration twelve the estimator reaches $0.024$ against MRFS's $0.091$ and
uniform's $0.113$, and on Blotto $0.082$ against $0.222$ and $0.293$. Both
margins are significant under a paired test over the sixteen seeds. Against MRFS,
$p = 0.022$ on the informative game and $p < 0.001$ on Blotto, and against uniform,
$p = 0.003$ and $p < 0.001$.

RWPS starts behind at iteration $0$ ($0.262$ against $0.180$ for MRFS), as expected with an
empty cache, and its lead then widens every iteration, which is what reusing earlier
simulations should produce.

Against the full rebuild the result divides by game. On the informative game the
estimator reaches $0.024$ against the rebuild's $0.036$ at a third of the
episodes, winning $13$ of $16$ seeds, but the paired difference is not
significant ($p = 0.057$). We claim only that $5\%$ of the cells cost no measurable quality. On Blotto the rebuild is
ahead, at $0.045$ against $0.082$ ($p = 0.004$).

IGS and the progressive-sampling family trail because both need more samples per profile
than the budget affords. A $5\%$ budget is ten cell evaluations at a $14\times14$ pool,
against the $294$ per cell that PSP's published schedule requires. The family matches RWPS
only at about $100\times$ the full matrix, where PS-REG-M first prunes (Table~\ref{tab:psp}), a
budget that defeats the purpose of estimating payoffs when every cell is a simulator rollout.

\begin{table}[t]
\centering\small
\setlength{\tabcolsep}{3pt}
\caption{RWPS against the progressive-sampling family on fixed $21\times21$ pools,
exploitability mean $\pm$ SD over four seeds at budgets that are multiples of the matrix.
Bold marks each row's best. The last column is the share of utility indices PS-REG-M
prunes, and the appendix gives all five budgets.}
\label{tab:psp}
\begin{tabular}{@{}llccccc@{}}
\toprule
Game & Budget & RWPS & GS & PSP & PS-REG-M & pruned \\
\midrule
\multirow{3}{*}{Inf.}
 & $1\times$   & $\mathbf{.029 \pm .023}$ & $.162 {\scriptstyle\pm} .138$ & $.222 {\scriptstyle\pm} .081$ & $.214 {\scriptstyle\pm} .106$ & $0$ \\
 & $20\times$  & $\mathbf{.012 \pm .006}$ & $.413 {\scriptstyle\pm} .125$ & $.156 {\scriptstyle\pm} .298$ & $.029 {\scriptstyle\pm} .050$ & $0$ \\
 & $100\times$ & $.007 {\scriptstyle\pm} .003$ & $.025 {\scriptstyle\pm} .040$ & $.005 {\scriptstyle\pm} .003$ & $\mathbf{.003 \pm .002}$ & $0.5\%$ \\
\midrule
\multirow{3}{*}{Blotto}
 & $1\times$   & $\mathbf{.043 \pm .011}$ & $.475 {\scriptstyle\pm} .029$ & $.471 {\scriptstyle\pm} .026$ & $.617 {\scriptstyle\pm} .268$ & $0$ \\
 & $20\times$  & $\mathbf{.025 \pm .004}$ & $.561 {\scriptstyle\pm} .081$ & $.366 {\scriptstyle\pm} .081$ & $.404 {\scriptstyle\pm} .049$ & $0$ \\
 & $100\times$ & $.015 {\scriptstyle\pm} .003$ & $.440 {\scriptstyle\pm} .035$ & $\mathbf{.013 \pm .001}$ & $.017 {\scriptstyle\pm} .007$ & $6.1\%$ \\
\bottomrule
\end{tabular}
\end{table}

\subsection{Budget sweeps on Cybersecurity Games}
\label{sec:cyber-sweep}

On both simulators RWPS reaches the lowest exploitability at the smallest budgets, and the
other arms close the gap only as the budget grows (Figures~\ref{fig:cygym-sweep}
and~\ref{fig:ansg-sweep}). Uniform sampling and MRFS need most of the available cells to reach the same floor, which
RWPS reaches with its first purchase. IGS stays above it throughout, and the
progressive-sampling family buys nothing at any budget shown.
RWPS leads from the first purchase because every arm starts from the same pre-cached
history but only RWPS learns from it, predicting every unbought cell and buying where the
equilibrium is sensitive and that prediction is least certain.

The two simulators separate the acquisition score from the fill differently. On
CyGym the uniform-acquisition ablation tracks the adaptive rule closely, so the
fill carries most of the margin. On ANSG the ablation is well above the floor at budgets 1--3 and reaches it
only at budget~4, three purchases after RWPS (adaptive), so the score makes the
first few purchases count. This is the cleanest test of the score in the paper, because
ANSG keeps every RWPS arm on the rank-one fill and the two variants differ only in which
cells they buy.
The appendix reports the same sweep on
Cyberwheel~\citep{oesch2024cyberwheel}, where the pool is flat and every buying arm ties.

\begin{figure}[t]
\centering
\includegraphics[width=0.72\linewidth]{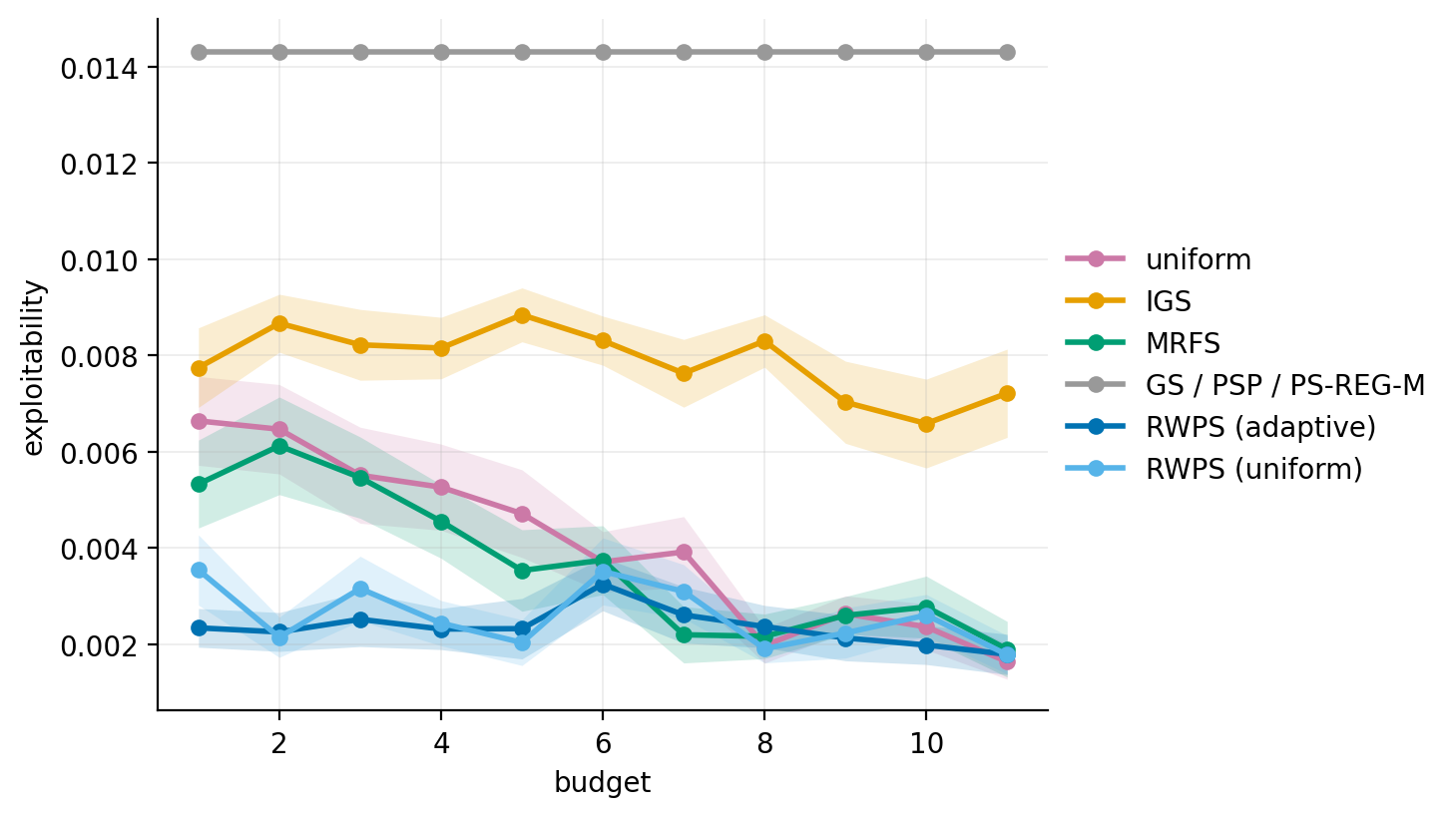}
\caption{CyGym budget sweep. $6\times6$ pool with a $5\times5$ history pre-cached for
every arm, leaving $11$ cells to buy; $16$ seeds. Exploitability is computed against a
reference matrix estimated at $16$ rollouts per cell. Bands are $\pm1$ standard error. GS/PSP/PS-REG-M is constant because its sampling schedule
requires more samples per cell than the budget affords, so it returns the history-only
estimate.}
\label{fig:cygym-sweep}
\end{figure}

\begin{figure}[t]
\centering
\includegraphics[width=0.72\linewidth]{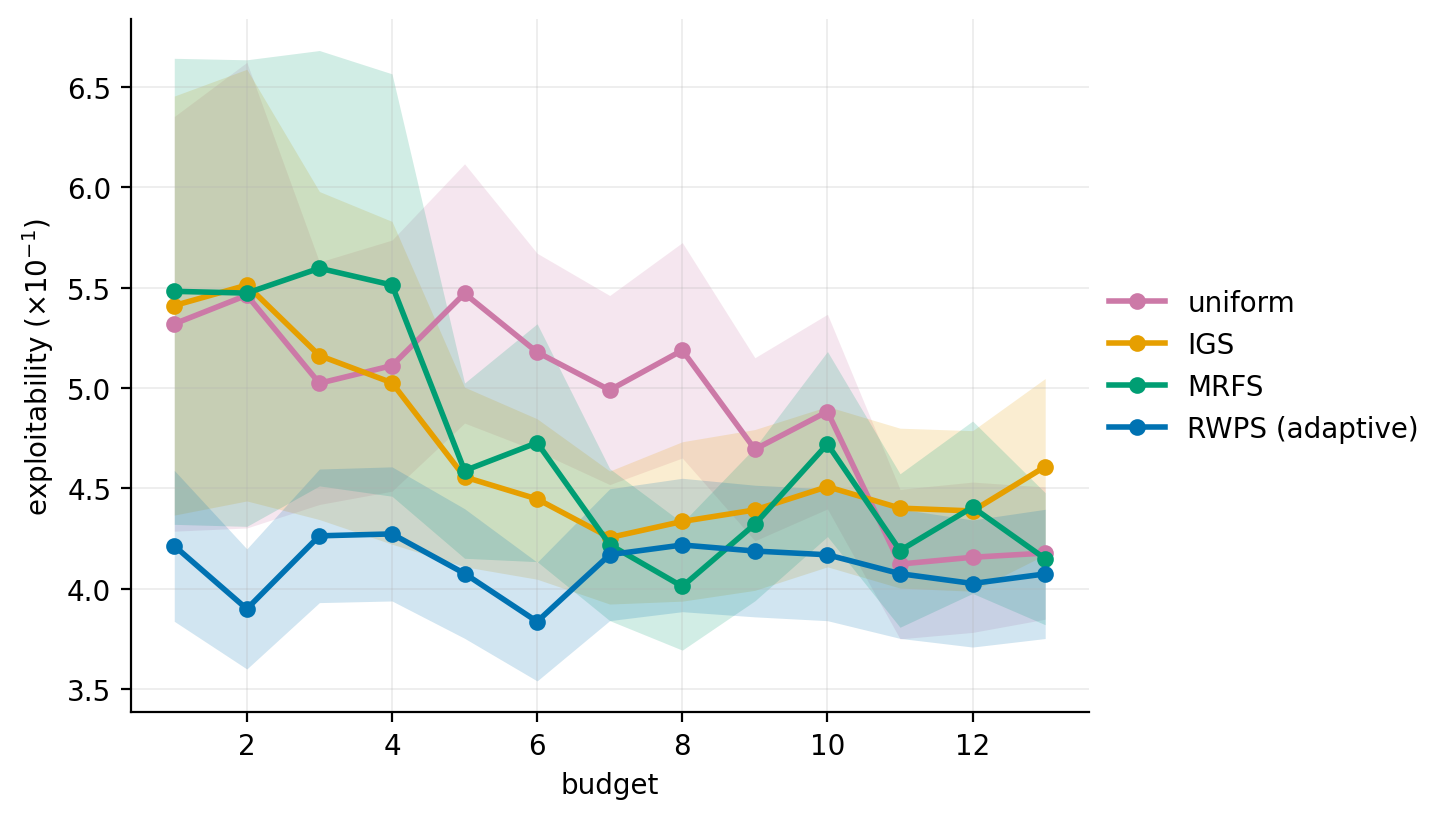}
\caption{ANSG budget sweep. $7\times7$ pool with a $6\times6$ history pre-cached for
every arm, leaving $13$ cells to buy, over $16$ seeds. Exploitability is computed exactly against
the fully simulated matrix of this pool, so no measurement is clamped. Bands are $\pm1$
standard error. Two arms are omitted for legibility. GS/PSP/PS-REG-M is constant at
$1.97\times10^{5}$ payoff units because its sampling schedule requires far more samples
per cell than the budget affords, so it buys nothing and returns the history-only
estimate. RWPS (uniform), the ablation with the acquisition score removed, is $2.11\times10^{4}$, $1.88\times10^{4}$ and
$1.57\times10^{4}$ at budgets 1--3, and reaches the floor only at budget 4, where
RWPS (adaptive) is already at it from budget 1.}
\label{fig:ansg-sweep}
\end{figure}

\subsection{Acquisition Ablation}
\label{sec:acq-ablation}

Table~\ref{tab:psro} varies the payoff build as a whole, so its margin mixes two
mechanisms, choosing which cells to simulate and predicting the ones we skip. Here we fix
the surrogate, budget, batch schedule and exploration rate and vary only the acquisition
score, on the same sixteen paired seeds, so any difference between rows is due to
acquisition alone.

Uniform here keeps the surrogate fill, while Table~\ref{tab:psro}'s replaces it with a
flat value, so comparing the two isolates the fill.

The fill accounts for $62\%$ of the total margin and the acquisition score for
$38\%$. On the informative game the flat-fill baseline reaches $0.113$, the surrogate-fill baseline $0.058$, and the full
method $0.024$. On Blotto the same three are $0.293$, $0.162$, and $0.082$. The
split holds on both games to the nearest percentage point.

\begin{table}[t]
\centering\small
\caption{Acquisition comparison. All rules share the surrogate, budget, batches and exploration
rate and differ only in the score. Mean $\pm$ one standard error of $\eps_{12}$ over sixteen paired seeds. Uncertainty sampling scores by surrogate disagreement
alone~\citep{settles2009active,srinivas2012gpucb}, and uniform selection keeps the fill
but drops the score. The appendix reports three further variants.}
\label{tab:acq}
\begin{tabular}{@{}lcc@{}}
\toprule
Acquisition score & Informative & Blotto \\
\midrule
$\nu$-smoothed product (ours)      & $\mathbf{0.024 \pm 0.004}$ & $\mathbf{0.082 \pm 0.011}$ \\
\citet{sokota2019deviation}        & $0.032 {\scriptstyle\pm} 0.004$ & $0.120 {\scriptstyle\pm} 0.011$ \\
uncertainty ($\hat\sigma$ only)    & $0.040 {\scriptstyle\pm} 0.007$ & $0.116 {\scriptstyle\pm} 0.013$ \\
uniform selection                  & $0.058 {\scriptstyle\pm} 0.016$ & $0.162 {\scriptstyle\pm} 0.019$ \\
\bottomrule
\end{tabular}
\end{table}

The smoothed rule is first on both games. On Blotto it beats
\citet{sokota2019deviation} by $0.038$ under a paired $t$-test ($p = 0.038$) and a
signed-rank test ($p = 0.016$), winning $13$ of $16$ seeds. On the informative game its
$0.008$ lead is not significant ($p = 0.224$), which we report as a tie.

Uncertainty sampling beats uniform selection on both games but loses to the full rule
($p = 0.021$ on Blotto, $p = 0.062$ on the informative game). Disagreement shows where the
fill is wrong, but a highly uncertain cell can still be irrelevant if no equilibrium
strategy plays it.

The appendix sweeps $\nu$, $\eps_{\mathrm{x}}$ and the bootstrap sizes.

\subsection{Regret Bound Analysis}
\label{sec:tightness}

\begin{figure}[t]
\centering
\includegraphics[width=0.6\linewidth]{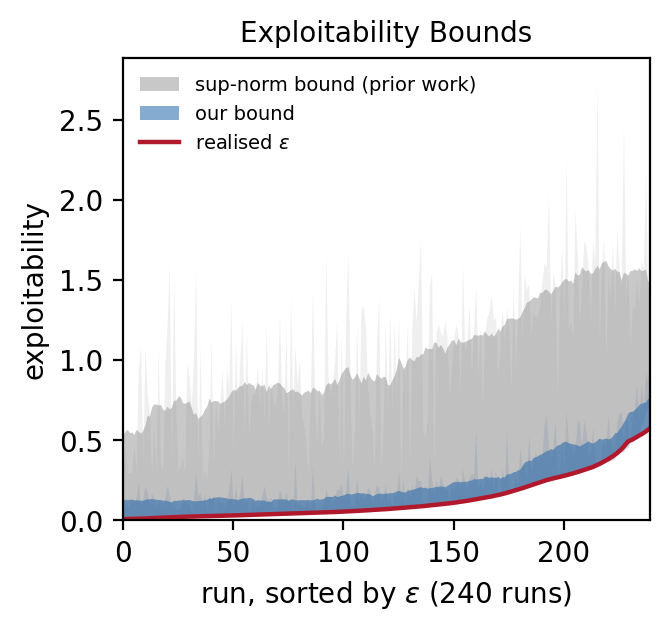}
\caption{Realised exploitability and its bounds on the estimator's own output,
over all $240$ fixed-pool budget-sweep runs on the three $21\times21$ games
(the informative and arbitrary-embedding latent-quality games and Colonel Blotto;
budgets $\{5,10,20,40,70\}\%$ of the matrix; $16$ seeds each), sorted by $\eps$. Shaded regions are each
bound's margin over $\eps$.}
\label{fig:perturb}
\end{figure}

Errors off the equilibrium support do not affect regret, and the estimator's own output
confirms it. Over all $240$ budget-sweep runs every bound holds on every run
(Figure~\ref{fig:perturb}). Our bound, the signed form (a) of Eq.~\eqref{eq:chain} plus the
measured solver term, averages $0.25$ of the sup-norm bound, with a median overestimation
factor of $2.3$ against $11.4$. Forms (c) and (b) average $0.47$ and $0.33$, as the chain
predicts. The estimator puts its error on cells no equilibrium strategy plays, those cells
set the sup-norm bound, and the support weighting ignores them. The solver term is needed,
since without it the signed bound fails on three Blotto runs, each by less than that run's
solver regret. These bounds are diagnostics and cannot be deployed as certificates. Each form
of Eq.~\eqref{eq:chain} contains $\Dhat - D$, so these builds use a fixed-pool bed where the
true matrix is available. The appendix gives per-game factors and the deployable
certificate.

\section{Conclusion}
\label{sec:conclusion}

In simulation-based security games the expensive step is measuring payoffs, and most of
those measurements can be avoided. On a $5\%$ per-build budget RWPS found less exploitable
equilibria than every other budgeted method we tested, and did best at the smallest budgets
on CyGym and ANSG. Our bounds are four to six times tighter than the standard one and need no
access to the true payoffs, so an analyst can state how exploitable a computed defence is
without ever simulating the full game. Together these make equilibrium analysis affordable at
the simulator fidelity the analysis calls for, including on new threat models such as
networks of LLM agents.

\FloatBarrier
\bibliographystyle{plainnat}
\bibliography{refs}

@inproceedings{standen2021cyborg,
  author    = {Standen, Maxwell and Lucas, Martin and Bowman, David and
               Richer, Toby J. and Kim, Junae and Marriott, Damian},
  title     = {{CybORG}: A Gym for the Development of Autonomous Cyber Agents},
  booktitle = {IJCAI-21 1st International Workshop on Adaptive Cyber Defense},
  year      = {2021}
}

@misc{cage2023,
  author       = {{TTCP CAGE Working Group}},
  title        = {Cyber Operations Research Gym},
  howpublished = {\url{https://github.com/cage-challenge/CybORG}},
  year         = {2022}
}

@article{emerson2024cyborgpp,
  author  = {Emerson, Harry and Zhou, Liz and Bowman, David and others},
  title   = {{CybORG++}: An Enhanced Gym for the Development of Autonomous
             Cyber Agents},
  journal = {arXiv preprint arXiv:2410.16324},
  year    = {2024}
}

@article{andrew2022yawningtitan,
  author  = {Andrew, Alex and Spillard, Sam and Collyer, Joshua and Dhir, Neil},
  title   = {Developing Optimal Causal Cyber-Defence Agents via Cyber Security
             Simulation},
  journal = {arXiv preprint arXiv:2207.12355},
  year    = {2022}
}

@misc{cyberbattlesim2021,
  author       = {{Microsoft Defender Research Team}},
  title        = {{CyberBattleSim}},
  howpublished = {\url{https://github.com/microsoft/cyberbattlesim}},
  year         = {2021}
}

@article{molina2021farland,
  author  = {Molina-Markham, Andres and Miniter, Cory and Powell, Becky and
             Ridley, Ahmad},
  title   = {Network Environment Design for Autonomous Cyberdefense},
  journal = {arXiv preprint arXiv:2103.07583},
  year    = {2021}
}

@inproceedings{lanier2026cygym,
  author    = {Lanier, Michael and Vorobeychik, Yevgeniy},
  title     = {{CyGym}: A Simulation-Based Game-Theoretic Analysis Framework
               for Cybersecurity},
  booktitle = {Proceedings of GameSec},
  series    = {LNCS 16223},
  publisher = {Springer},
  year      = {2026}
}

@inproceedings{lanier2026metadoar,
  author    = {Lanier, Michael and Vorobeychik, Yevgeniy},
  title     = {A Scalable Approach to Solving Simulation-Based Network Security
               Games},
  booktitle = {Proceedings of FLAIRS},
  year      = {2026}
}

@inproceedings{lanctot2017psro,
  author    = {Lanctot, Marc and Zambaldi, Vinicius and Gruslys, Audrunas and
               Lazaridou, Angeliki and Tuyls, Karl and P{\'e}rolat, Julien and
               Silver, David and Graepel, Thore},
  title     = {A Unified Game-Theoretic Approach to Multiagent Reinforcement
               Learning},
  booktitle = {Advances in Neural Information Processing Systems (NeurIPS)},
  year      = {2017}
}

@inproceedings{mcmahan2003double,
  author    = {McMahan, H. Brendan and Gordon, Geoffrey J. and Blum, Avrim},
  title     = {Planning in the Presence of Cost Functions Controlled by an
               Adversary},
  booktitle = {International Conference on Machine Learning (ICML)},
  year      = {2003}
}

@article{vorobeychik2007payoff,
  author  = {Vorobeychik, Yevgeniy and Wellman, Michael P. and Singh, Satinder},
  title   = {Learning Payoff Functions in Infinite Games},
  journal = {Machine Learning},
  volume  = {67},
  number  = {1--2},
  pages   = {145--168},
  year    = {2007}
}

@article{wellman2006egta,
  author  = {Wellman, Michael P.},
  title   = {Methods for Empirical Game-Theoretic Analysis},
  journal = {Proceedings of the AAAI Conference on Artificial Intelligence},
  year    = {2006}
}

@inproceedings{rowland2019multiagent,
  author    = {Rowland, Mark and Omidshafiei, Shayegan and Tuyls, Karl and
               P{\'e}rolat, Julien and Valko, Michal and Piliouras, Georgios and
               Munos, R{\'e}mi},
  title     = {Multiagent Evaluation under Incomplete Information},
  booktitle = {Advances in Neural Information Processing Systems (NeurIPS)},
  year      = {2019}
}

@article{omidshafiei2019alpharank,
  author  = {Omidshafiei, Shayegan and Papadimitriou, Christos and
             Piliouras, Georgios and Tuyls, Karl and Rowland, Mark and
             Lespiau, Jean-Baptiste and Czarnecki, Wojciech M. and
             Lanctot, Marc and P{\'e}rolat, Julien and Munos, R{\'e}mi},
  title   = {$\alpha$-Rank: Multi-Agent Evaluation by Evolution},
  journal = {Scientific Reports},
  volume  = {9},
  year    = {2019}
}

@inproceedings{balduzzi2019openended,
  author    = {Balduzzi, David and Garnelo, Marta and Bachrach, Yoram and
               Czarnecki, Wojciech M. and P{\'e}rolat, Julien and
               Jaderberg, Max and Graepel, Thore},
  title     = {Open-ended Learning in Symmetric Zero-sum Games},
  booktitle = {International Conference on Machine Learning (ICML)},
  year      = {2019}
}

@inproceedings{czarnecki2020spinningtops,
  author    = {Czarnecki, Wojciech M. and Gidel, Gauthier and Tracey, Brendan and
               Tuyls, Karl and Omidshafiei, Shayegan and Balduzzi, David and
               Jaderberg, Max},
  title     = {Real World Games Look Like Spinning Tops},
  booktitle = {Advances in Neural Information Processing Systems (NeurIPS)},
  year      = {2020}
}

@article{brown1951fictitious,
  author  = {Brown, George W.},
  title   = {Iterative Solution of Games by Fictitious Play},
  journal = {Activity Analysis of Production and Allocation},
  year    = {1951}
}

@inproceedings{lakshminarayanan2017ensembles,
  author    = {Lakshminarayanan, Balaji and Pritzel, Alexander and
               Blundell, Charles},
  title     = {Simple and Scalable Predictive Uncertainty Estimation using Deep
               Ensembles},
  booktitle = {Advances in Neural Information Processing Systems (NeurIPS)},
  year      = {2017}
}

@article{settles2009active,
  author  = {Settles, Burr},
  title   = {Active Learning Literature Survey},
  journal = {University of Wisconsin-Madison Department of Computer Sciences
             Technical Report 1648},
  year    = {2009}
}

@article{shahriari2016bayesopt,
  author  = {Shahriari, Bobak and Swersky, Kevin and Wang, Ziyu and
             Adams, Ryan P. and de Freitas, Nando},
  title   = {Taking the Human Out of the Loop: A Review of Bayesian
             Optimization},
  journal = {Proceedings of the IEEE},
  volume  = {104},
  number  = {1},
  pages   = {148--175},
  year    = {2016}
}

@inproceedings{tuyls2018generalised,
  author    = {Tuyls, Karl and P{\'e}rolat, Julien and Lanctot, Marc and
               Leibo, Joel Z. and Graepel, Thore},
  title     = {A Generalised Method for Empirical Game Theoretic Analysis},
  booktitle = {Proceedings of the 17th International Conference on Autonomous
               Agents and MultiAgent Systems (AAMAS)},
  year      = {2018}
}

@article{tuyls2020bounds,
  author  = {Tuyls, Karl and P{\'e}rolat, Julien and Lanctot, Marc and
             Hughes, Edward and Everett, Richard and Leibo, Joel Z. and
             Szepesv{\'a}ri, Csaba and Graepel, Thore},
  title   = {Bounds and Dynamics for Empirical Game Theoretic Analysis},
  journal = {Autonomous Agents and Multi-Agent Systems},
  volume  = {34},
  number  = {7},
  year    = {2020}
}

@article{brafman2002rmax,
  author  = {Brafman, Ronen I. and Tennenholtz, Moshe},
  title   = {R-MAX: A General Polynomial Time Algorithm for Near-Optimal
             Reinforcement Learning},
  journal = {Journal of Machine Learning Research},
  volume  = {3},
  pages   = {213--231},
  year    = {2002}
}

@inproceedings{areyanviqueira2019learning,
  author    = {Areyan Viqueira, Enrique and Greenwald, Amy and
               Cousins, Cyrus and Upfal, Eli},
  title     = {Learning Simulation-Based Games from Data},
  booktitle = {Proceedings of the 18th International Conference on
               Autonomous Agents and MultiAgent Systems (AAMAS)},
  pages     = {1778--1780},
  year      = {2019}
}

@inproceedings{areyanviqueira2020improved,
  author    = {Areyan Viqueira, Enrique and Cousins, Cyrus and Greenwald, Amy},
  title     = {Improved Algorithms for Learning Equilibria in
               Simulation-Based Games},
  booktitle = {Proceedings of the 19th International Conference on
               Autonomous Agents and MultiAgent Systems (AAMAS)},
  pages     = {79--87},
  year      = {2020}
}

@article{cousins2022computational,
  author  = {Cousins, Cyrus and Mishra, Bhaskar and
             Areyan Viqueira, Enrique and Greenwald, Amy},
  title   = {Computational and Data Requirements for Learning Generic
             Properties of Simulation-Based Games},
  journal = {arXiv preprint arXiv:2208.06400},
  year    = {2022}
}

@misc{mishra2023regretpruning,
  author       = {Mishra, Bhaskar and Cousins, Cyrus and Greenwald, Amy},
  title        = {Regret Pruning for Learning Equilibria in Simulation-Based Games},
  year         = {2022},
  howpublished = {arXiv preprint arXiv:2211.16670}
}

@article{picheny2016bayesian,
  author  = {Picheny, Victor and Binois, Mickael and Habbal, Abderrahmane},
  title   = {A {B}ayesian Optimization Approach to Find {N}ash Equilibria},
  journal = {Journal of Global Optimization},
  volume  = {73},
  number  = {1},
  pages   = {171--192},
  year    = {2019},
  doi     = {10.1007/s10898-018-0688-0}
}

@inproceedings{sokota2019deviation,
  author    = {Sokota, Samuel and Ho, Caleb and Wiedenbeck, Bryce},
  title     = {Learning Deviation Payoffs in Simulation-Based Games},
  booktitle = {Proceedings of the AAAI Conference on Artificial Intelligence},
  year      = {2019}
}

@inproceedings{vorobeychik2008stochastic,
  author    = {Vorobeychik, Yevgeniy and Wellman, Michael P.},
  title     = {Stochastic Search Methods for {N}ash Equilibrium Approximation
               in Simulation-Based Games},
  booktitle = {Proceedings of the 7th International Joint Conference on
               Autonomous Agents and Multiagent Systems (AAMAS)},
  pages     = {1055--1062},
  year      = {2008}
}

@inproceedings{jordan2008mrfs,
  author    = {Jordan, Patrick R. and Vorobeychik, Yevgeniy and Wellman, Michael P.},
  title     = {Searching for Approximate Equilibria in Empirical Games},
  booktitle = {Proceedings of the 7th International Joint Conference on
               Autonomous Agents and Multiagent Systems (AAMAS)},
  pages     = {1063--1070},
  year      = {2008}
}

@inproceedings{wiedenbeck2012dpr,
  author    = {Wiedenbeck, Bryce and Wellman, Michael P.},
  title     = {Scaling Simulation-Based Game Analysis through
               Deviation-Preserving Reduction},
  booktitle = {Proceedings of the 11th International Conference on
               Autonomous Agents and Multiagent Systems (AAMAS)},
  pages     = {931--938},
  year      = {2012}
}

@inproceedings{wiedenbeck2014bootstrap,
  author    = {Wiedenbeck, Bryce and Cassell, Ben-Alexander and Wellman, Michael P.},
  title     = {Bootstrap Statistics for Empirical Games},
  booktitle = {Proceedings of the 13th International Conference on
               Autonomous Agents and Multiagent Systems (AAMAS)},
  pages     = {597--604},
  year      = {2014}
}

@article{srinivas2012gpucb,
  author  = {Srinivas, Niranjan and Krause, Andreas and Kakade, Sham M. and
             Seeger, Matthias W.},
  title   = {Information-Theoretic Regret Bounds for {Gaussian} Process
             Optimization in the Bandit Setting},
  journal = {IEEE Transactions on Information Theory},
  volume  = {58},
  number  = {5},
  pages   = {3250--3265},
  year    = {2012}
}

@article{howard2021timeuniform,
  author  = {Howard, Steven R. and Ramdas, Aaditya and McAuliffe, Jon and
             Sekhon, Jasjeet},
  title   = {Time-Uniform, Nonparametric, Nonasymptotic Confidence Sequences},
  journal = {The Annals of Statistics},
  volume  = {49},
  number  = {2},
  pages   = {1055--1080},
  year    = {2021}
}

@article{wellman2025egtasurvey,
  title   = {Empirical Game-Theoretic Analysis: A Survey},
  author  = {Wellman, Michael P. and Tuyls, Karl and Greenwald, Amy},
  journal = {Journal of Artificial Intelligence Research},
  year    = {2025},
  note    = {arXiv:2403.04018}
}

@inproceedings{wang2023mfgegta,
  title     = {Empirical Game-Theoretic Analysis for Mean Field Games},
  author    = {Wang, Yongzhao and Ma, Qiurui and Wellman, Michael P.},
  booktitle = {Proceedings of the 22nd International Conference on Autonomous Agents and Multiagent Systems (AAMAS)},
  pages     = {1025--1033},
  year      = {2023}
}

@inproceedings{wang2022strategyexploration,
  title     = {Evaluating Strategy Exploration in Empirical Game-Theoretic Analysis},
  author    = {Wang, Yongzhao and Ma, Qiurui and Wellman, Michael P.},
  booktitle = {Proceedings of the 21st International Conference on Autonomous Agents and Multiagent Systems (AAMAS)},
  pages     = {1346--1354},
  year      = {2022}
}

@inproceedings{oesch2024cyberwheel,
  author    = {Oesch, Sean and Chaulagain, Amul and Weber, Brian and Dixson, Matthew and Sadovnik, Amir and Roberson, Benjamin and Watson, Cory and Austria, Phillipe},
  title     = {Towards a High Fidelity Training Environment for Autonomous Cyber Defense Agents},
  booktitle = {Proceedings of the 17th Cyber Security Experimentation and Test Workshop},
  series    = {CSET '24},
  pages     = {91--99},
  publisher = {Association for Computing Machinery},
  year      = {2024},
  doi       = {10.1145/3675741.3675752}
}

@inproceedings{rashid2021alpharank,
  author    = {Rashid, Tabish and Zhang, Cheng and Ciosek, Kamil},
  title     = {Estimating $\alpha$-Rank by Maximizing Information Gain},
  booktitle = {Proceedings of the AAAI Conference on Artificial Intelligence},
  volume    = {35},
  pages     = {5673--5681},
  year      = {2021},
  doi       = {10.1609/aaai.v35i6.16712}
}

@misc{savani2025gambit,
  author       = {Savani, Rahul and Turocy, Theodore L.},
  title        = {Gambit: The Package for Computation in Game Theory, Version 16.2.1},
  year         = {2025},
  howpublished = {\url{https://www.gambit-project.org}}
}

@inproceedings{han2024arise,
  author    = {Han, Minbiao and Zhang, Fengxue and Chen, Yuxin},
  title     = {No-Regret Learning of {N}ash Equilibrium for Black-Box Games via {G}aussian Processes},
  booktitle = {Proceedings of the 40th Conference on Uncertainty in Artificial Intelligence (UAI)},
  year      = {2024},
  url       = {https://arxiv.org/abs/2405.08318}
}

@inproceedings{aldujaili2018blackbox,
  author    = {Al-Dujaili, Abdullah and Hemberg, Erik and O'Reilly, Una-May},
  title     = {Approximating {N}ash Equilibria for Black-Box Games: A {B}ayesian Optimization Approach},
  booktitle = {Workshop on Optimization in Multiagent Systems (OptMAS)},
  year      = {2018},
  url       = {https://arxiv.org/abs/1804.10586}
}

@inproceedings{nguyen2026conservative,
  author    = {Nguyen, Austin A. and Wellman, Michael P.},
  title     = {Conservative Equilibrium Discovery in Offline Game-Theoretic Multiagent Reinforcement Learning},
  booktitle = {Proceedings of the 25th International Conference on Autonomous Agents and MultiAgent Systems (AAMAS)},
  year      = {2026},
  url       = {https://arxiv.org/abs/2603.00374}
}

@inproceedings{hammar2020finding,
  author    = {Hammar, Kim and Stadler, Rolf},
  title     = {Finding Effective Security Strategies through Reinforcement Learning and Self-Play},
  booktitle = {Proceedings of the 16th International Conference on Network and Service Management (CNSM)},
  pages     = {1--9},
  publisher = {IEEE},
  year      = {2020}
}

@inproceedings{kunz2022multiagent,
  author    = {Kunz, Thomas and Fisher, Christian and La Novara-Gsell, James and Nguyen, Christopher and Li, Li},
  title     = {A Multiagent {CyberBattleSim} for {RL} Cyber Operation Agents},
  booktitle = {Proceedings of the International Conference on Computational Science and Computational Intelligence (CSCI)},
  pages     = {897--903},
  publisher = {IEEE},
  year      = {2022}
}

@misc{schwartz2019autonomous,
  author       = {Schwartz, Jonathon and Kurniawati, Hanna},
  title        = {Autonomous Penetration Testing using Reinforcement Learning},
  year         = {2019},
  howpublished = {arXiv:1905.05965}
}

@inproceedings{hammar2026csle,
  author    = {Hammar, Kim},
  title     = {{CSLE}: A Reinforcement Learning Platform for Autonomous Security Management},
  booktitle = {Proceedings of the Ninth Annual Conference on Machine Learning and Systems (MLSys)},
  year      = {2026}
}

\clearpage
\appendix
\section{Appendix}
\label{app:root}

This appendix collects the proofs, the specification of the ANSG domain, implementation
details and hyperparameters, and additional experiments.

\subsection{Proofs}
\label{app:proofs}

This appendix collects full proofs of the results in Section~\ref{sec:theory}, in order of
appearance. Throughout, rows index defender strategies $d_1,\dots,d_{n_D}$ and columns
attacker strategies $a_1,\dots,a_{n_A}$; $D, \Dhat \in \mathbb{R}^{n_D\times n_A}$ hold the
defender's true and estimated payoffs, and $A, \Ahat \in \mathbb{R}^{n_A\times n_D}$ the
attacker's, indexed attacker-first. For a mixed profile $(p,q)$ the defender's regret in a
matrix $M$ is
\[
\mathrm{reg}_D(p,q;M) \;=\; \max_i\, e_i^\top M q \;-\; p^\top M q ,
\]
and symmetrically $\mathrm{reg}_A(p,q;N) = \max_j e_j^\top N p - q^\top N p$ for the
attacker. A profile is an $\eps$-Nash equilibrium of $(D,A)$ when both regrets are at most
$\eps$, and its exploitability is $\eps = \max\{\mathrm{reg}_D, \mathrm{reg}_A\}$. That
$(\hat p,\hat q)$ is an $\eps_{\mathrm{solve}}$-approximate equilibrium of $(\Dhat,\Ahat)$
means
\begin{equation}
\label{eq:solve}
\max_i e_i^\top \Dhat \hat q \;\le\; \hat p^\top \Dhat \hat q + \eps_{\mathrm{solve}},
\qquad
\max_j e_j^\top \Ahat \hat p \;\le\; \hat q^\top \Ahat \hat p + \eps_{\mathrm{solve}}.
\end{equation}
Write $\Delta_D = \Dhat - D$ and $\Delta_A = \Ahat - A$, and $|\cdot|$ for the entrywise
absolute value.

\begin{proof}[Proof of Lemma~\ref{lem:supp}]
We prove the four links of Eq.~\eqref{eq:chain} for the defender in turn, then the
attacker's analogue and the consequence for $\eps$.

We begin with form (a). Fix a pure deviation $i$. Since $\Dhat = D + \Delta_D$,
\[
e_i^\top D \hat q \;=\; e_i^\top \Dhat \hat q - e_i^\top \Delta_D \hat q .
\]
By Eq.~\eqref{eq:solve}, $e_i^\top \Dhat \hat q \le \hat p^\top \Dhat \hat q + \eps_{\mathrm{solve}}$,
and expanding $\Dhat$ once more,
$\hat p^\top \Dhat \hat q = \hat p^\top D \hat q + \hat p^\top \Delta_D \hat q$. Combining,
\[
e_i^\top D \hat q - \hat p^\top D \hat q
\;\le\; \eps_{\mathrm{solve}} + \hat p^\top \Delta_D \hat q + \bigl(-\Delta_D \hat q\bigr)_i .
\]
The right-hand side depends on $i$ only through the last term, so taking the maximum over $i$
on both sides gives
$\mathrm{reg}_D - \eps_{\mathrm{solve}} \le \max_i(-\Delta_D\hat q)_i + \hat p^\top \Delta_D \hat q$,
which is (a). No property of the game other than bilinearity has been used; in particular
the argument does not assume zero sum.

To pass from (a) to (b), we bound the two terms separately. For every row $i$,
$(-\Delta_D \hat q)_i = \sum_j \hat q_j (-\Delta_{D,ij}) \le \sum_j \hat q_j |\Delta_{D,ij}|
= (|\Delta_D|\hat q)_i$, because $\hat q_j \ge 0$ and $-x \le |x|$; taking the maximum over
$i$ gives $\max_i(-\Delta_D\hat q)_i \le \tau_D(\hat q)$. Likewise
$\hat p^\top \Delta_D \hat q = \sum_{i,j}\hat p_i \hat q_j \Delta_{D,ij}
\le \sum_{i,j}\hat p_i \hat q_j |\Delta_{D,ij}| = \hat p^\top |\Delta_D| \hat q$, since all weights
are non-negative.

To pass from (b) to (c), write $\hat p^\top |\Delta_D|\hat q = \sum_i \hat p_i (|\Delta_D|\hat q)_i$.
This is a convex combination of the numbers $(|\Delta_D|\hat q)_i$, since $\hat p$ is a
probability vector, and a convex combination never exceeds its largest element. Hence
$\hat p^\top |\Delta_D|\hat q \le \max_i (|\Delta_D|\hat q)_i = \tau_D(\hat q)$, and
(b) $\le 2\tau_D(\hat q)$.

Finally, for (c) to (d), note that for every row,
$(|\Delta_D|\hat q)_i = \sum_j \hat q_j |\Delta_{D,ij}| \le \|\Delta_D\|_\infty \sum_j \hat q_j
= \|\Delta_D\|_\infty \le \tau$, so $\tau_D(\hat q) \le \tau$. The sum ranges only over
$j \in \mathrm{supp}(\hat q)$: an entry in a column of zero mass contributes nothing, however
large. For the equality condition, suppose $\tau_D(\hat q) = \tau$. Then some row $i$
attains $\sum_j \hat q_j |\Delta_{D,ij}| = \tau$; since each term satisfies
$|\Delta_{D,ij}| \le \tau$ and the weights sum to one, this forces $|\Delta_{D,ij}| = \tau$ for
every $j$ with $\hat q_j > 0$. So a worst-error cell lies in a played column, which is the
stated necessary condition.

The same argument covers the attacker, whose payoff to pure strategy $j$ against $\hat p$ is
$e_j^\top A \hat p$; the four steps apply with the roles exchanged: replace
$(D,\Dhat,\hat p,\hat q,i)$ by $(A,\Ahat,\hat q,\hat p,j)$ and use the second inequality of
Eq.~\eqref{eq:solve}. This gives the attacker's chain with
$\tau_A(\hat p) = \max_j (|\Delta_A|\hat p)_j$ in place of $\tau_D(\hat q)$; played rows now
play the role of played columns.

Each form bounds its player's regret minus $\eps_{\mathrm{solve}}$. The
exploitability is the larger of the two regrets, so it is at most $\eps_{\mathrm{solve}}$
plus the larger of the two players' values of whichever form is used, which makes
$(\hat p,\hat q)$ an $\eps$-Nash equilibrium of $(D,A)$ for that $\eps$.
\end{proof}

\begin{proof}[Proof of Theorem~\ref{thm:eps}]
The argument has three steps: a confidence event under which every simulated entry is
accurate, a split of each row's weighted error on that event, and an application of
Lemma~\ref{lem:supp}.

We first build a confidence event under which every simulated entry is accurate. Fix one
payoff entry, say $D_{ij}$, and let
$\bar X_n$ be the average of its first $n$ rollouts, each equal to $D_{ij}$ plus independent
$\sigma$-sub-Gaussian noise. A stitched time-uniform boundary of the kind developed
by~\citet{howard2021timeuniform} gives, for any $\delta' \in (0,1)$,
\begin{align*}
\Pr\Bigl[\,&\exists\, n \ge 1:\\
&\bigl|\bar X_n - D_{ij}\bigr| >
\sigma\sqrt{\tfrac{2(\log(2/\delta') + \log\log_2(2n))}{n}}\,\Bigr] \;\le\; \delta' ,
\end{align*}
where the factor $2$ inside the logarithm accounts for the two tails. The guarantee holds
for all $n$ simultaneously, and therefore also at any random visit count, including one
chosen by an acquisition rule on the basis of earlier samples; this is what makes it valid
under adaptive sampling, where a fixed-$n$ bound would not be. There are $n_D n_A$ entries
in $D$ and $n_D n_A$ in $A$. Setting $\delta' = \delta/(2n_Dn_A)$ and taking a union bound,
the event
\begin{align*}
\mathcal{E}:\quad &\bigl|\Dhat_{ij} - D_{ij}\bigr| \le \zeta(m_{ij},\delta)
\ \text{ and }\ \bigl|\Ahat_{ji} - A_{ji}\bigr| \le \zeta(m_{ij},\delta)\\
&\text{for every simulated cell } (i,j)
\end{align*}
holds with probability at least $1-\delta$, where $m_{ij}$ is the cell's visit count and
$\log(2/\delta') = \log(4n_Dn_A/\delta)$ as in $\zeta$. Finally, $\zeta(\cdot,\delta)$ is
non-increasing in the visit count: the numerator grows only doubly logarithmically, while
the denominator grows linearly. Since every simulated cell has $m_{ij} \ge m$, on
$\mathcal{E}$ every simulated entry has error at most $\zeta(m,\delta)$.

Next we split each row's weighted error on that event. Fix a row $i$, and let $S_i$ and $U_i$ be the columns
whose cells in that row are simulated and surrogate-filled, respectively. On $\mathcal{E}$,
\begin{align*}
(|\Delta_D|\hat q)_i
&= \sum_{j\in S_i} \hat q_j |\Delta_{D,ij}| + \sum_{j\in U_i} \hat q_j |\Delta_{D,ij}|\\
&\le \zeta(m,\delta)\sum_{j\in S_i}\hat q_j \;+\; \beta \sum_{j\in U_i} \hat q_j ,
\end{align*}
using Step 1 on the first sum and the assumption that surrogate entries have absolute error
at most $\beta$ on the second. The first sum of weights is at most one. The second is the
$\hat q$-mass row $i$ places on surrogate cells, which is at most its maximum over rows,
$w_{\mathrm{sur}}(\hat q)$. Hence $(|\Delta_D|\hat q)_i \le \zeta(m,\delta) + \beta\,
w_{\mathrm{sur}}(\hat q)$ for every $i$, and taking the maximum,
$\tau_D(\hat q) \le \zeta(m,\delta) + \beta\, w_{\mathrm{sur}}(\hat q)$. The attacker's
version, $\tau_A(\hat p) \le \zeta(m,\delta) + \beta\, w_{\mathrm{sur}}(\hat p)$, follows in
the same way over the attacker's rows.

It remains to apply Lemma~\ref{lem:supp}. Form (c) of Lemma~\ref{lem:supp} gives
$\mathrm{reg}_D \le \eps_{\mathrm{solve}} + 2\tau_D(\hat q)$ and
$\mathrm{reg}_A \le \eps_{\mathrm{solve}} + 2\tau_A(\hat p)$. Substituting Step 2 and taking the
larger of the two regrets yields the stated bound on $\eps$, on the event $\mathcal{E}$ of
probability at least $1-\delta$. Because $\mathcal{E}$ was built to hold at every visit
count, the bound holds for whatever cells and visit counts the acquisition rule chose.
\end{proof}

\begin{proof}[Proof of Corollary~\ref{cor:cert}]
Work on the event $\mathcal{E}$ of the previous proof, which has probability at least
$1-\delta$, and consider the defender; the attacker is symmetric.

We first bound the best deviation from above, and claim that $D \le D^{\mathrm{opt}} + \zeta_{m,\delta}$
entrywise. On a simulated cell, $D^{\mathrm{opt}}_{ij}$ is the measured mean, which on
$\mathcal{E}$ is within $\zeta_{m,\delta}$ of $D_{ij}$. On an unsimulated cell,
$D^{\mathrm{opt}}_{ij}$ is the maximum payoff, which bounds $D_{ij}$ by definition. Since
$\hat q \ge 0$ and sums to one, it follows that for every $i$,
$e_i^\top D \hat q \le e_i^\top D^{\mathrm{opt}}\hat q + \zeta_{m,\delta}$, and so
$\max_i e_i^\top D\hat q \le \max_i (D^{\mathrm{opt}}\hat q)_i + \zeta_{m,\delta}$.

Next we bound the equilibrium value from below. The value term needs the matrix the solver was
given to be accurate where the equilibrium plays, that is, on the support block
$\mathrm{supp}(\hat p)\times\mathrm{supp}(\hat q)$. When that block is simulated, every term
of $\hat p^\top \Dhat\hat q = \sum_{i,j}\hat p_i\hat q_j\Dhat_{ij}$ with non-zero weight is a
measured mean, so on $\mathcal{E}$,
$\hat p^\top D \hat q \ge \hat p^\top \Dhat \hat q - \zeta_{m,\delta}$.

Subtracting the second bound from the first,
$\mathrm{reg}_D = \max_i e_i^\top D\hat q - \hat p^\top D\hat q
\le \max_i (D^{\mathrm{opt}}\hat q)_i - \hat p^\top \Dhat \hat q + 2\zeta_{m,\delta}$, which is at
most the stated bound; the $\eps_{\mathrm{solve}}$ term is slack and is kept only so that the
certificate dominates Theorem~\ref{thm:eps} term by term. The same argument with
$A^{\mathrm{opt}}$ and the rows of $\mathrm{supp}(\hat p)$ bounds $\mathrm{reg}_A$, and the
exploitability is the larger of the two.

The condition on the support block cannot simply be dropped. If some support cell is surrogate-filled, the value term
can overstate $\hat p^\top D \hat q$ by up to $\beta$ times its weight, and that error is not
visible in the data. Replacing $\Dhat$ in the value term by a pessimistic fill, measured
values on simulated cells and the minimum payoff elsewhere, restores the bound without any
condition, at the price of a looser value term.
\end{proof}

\begin{proof}[Proof of Corollary~\ref{cor:coverage}]
We first show that both multipliers vanish. By definition,
\[
w_{\mathrm{sur}}(\hat q) = \max_i \sum_{j} \hat q_j\,\mathbf{1}[(i,j)\text{ surrogate}],
\]
and only columns $j \in \mathrm{supp}(\hat q)$ carry weight. The hypothesis says that every cell
in those columns is simulated, in every row, so for each row the sum has no non-zero term
and $w_{\mathrm{sur}}(\hat q) = 0$. The same argument over the rows of
$\mathrm{supp}(\hat p)$ gives $w_{\mathrm{sur}}(\hat p) = 0$.

Substituting $w_{\mathrm{sur}} = 0$ into Theorem~\ref{thm:eps}
removes the term $\beta\, w_{\mathrm{sur}}$ for both players, leaving
$\eps \le \eps_{\mathrm{solve}} + 2\zeta(m,\delta)$. Since $\beta$ appears only in the removed
term, the bound holds whatever the surrogate's error.

It remains to count the cells. The set is the union of $C = \{1{:}n_D\}\times\mathrm{supp}(\hat q)$, with
$n_D s_A$ cells, and $R = \mathrm{supp}(\hat p)\times\{1{:}n_A\}$, with $n_A s_D$ cells. Their
intersection is $\mathrm{supp}(\hat p)\times\mathrm{supp}(\hat q)$, with $s_D s_A$ cells, so by
inclusion and exclusion $|C\cup R| = n_D s_A + n_A s_D - s_D s_A = B^\ast$.
\end{proof}

\begin{proof}[Proof of Proposition~\ref{prop:coverage}]
Coverage is monotone: a cell once simulated is never returned to the surrogate, so the set
of cached cells is non-decreasing in $t$.

Let $\mathcal{R}_t$ denote the set consisting of every row simulated in all $n_A$ columns
together with every column simulated in all $n_D$ rows, after round $t$. Round $t$ simulates
$\{1{:}n_D\}\times\mathrm{supp}(\hat q_t)$, which completes every column of
$\mathrm{supp}(\hat q_t)$, and $\mathrm{supp}(\hat p_t)\times\{1{:}n_A\}$, which completes
every row of $\mathrm{supp}(\hat p_t)$; hence
$\mathrm{supp}(\hat p_t) \cup \mathrm{supp}(\hat q_t) \subseteq \mathcal{R}_t$.

Suppose the procedure does not halt after round $t+1$, so
$w_{\mathrm{sur}}(\hat q_{t+1}) > 0$ or $w_{\mathrm{sur}}(\hat p_{t+1}) > 0$. In the first
case there are $i$ and $j$ with $\hat q_{t+1,j} > 0$ and $(i,j)$ unsimulated, so column $j$
is incomplete and therefore $j \notin \mathrm{supp}(\hat q_t)$; round $t+1$ simulates
$\{1{:}n_D\}\times\mathrm{supp}(\hat q_{t+1}) \ni (\cdot,j)$ and so adds $j$ to
$\mathcal{R}$. The second case is symmetric in rows. Every non-halting round therefore
enlarges $\mathcal{R}$ by at least one element, and $|\mathcal{R}| \le n_D + n_A$, which
gives (i).

For (ii), halting means $w_{\mathrm{sur}}(\hat p) = w_{\mathrm{sur}}(\hat q) = 0$ at the
mixture returned by the final solve, which is exactly the hypothesis of
Corollary~\ref{cor:coverage} evaluated at that mixture; the conclusion of
Theorem~\ref{thm:eps} then loses its $\beta$ term. For (iii), monotonicity bounds the total
simulated set by the whole matrix.
\end{proof}

\begin{proof}[Proof of the per-iteration accounting above]
Write $\hat p, \hat q$ for the unchanged supports and $d_{\mathrm{new}}, a_{\mathrm{new}}$
for the added strategies. The set
$\{1{:}n_D\}\times\mathrm{supp}(\hat q)$ acquires one new row, contributing the $s_A$ cells
$(d_{\mathrm{new}}, j)$ for $j \in \mathrm{supp}(\hat q)$; the set
$\mathrm{supp}(\hat p)\times\{1{:}n_A\}$ acquires one new column, contributing the $s_D$
cells $(i, a_{\mathrm{new}})$ for $i \in \mathrm{supp}(\hat p)$. The cell
$(d_{\mathrm{new}}, a_{\mathrm{new}})$ lies in neither set, since by hypothesis neither new
strategy is in a support. All other cells of the coverage set were simulated at the previous
iteration and are cached. A cached full rebuild instead simulates every newly exposed cell,
of which there are $(k+1)^2 - k^2 = 2k+1$.
\end{proof}

\begin{proof}[Proof of Corollary~\ref{cor:budget}]
Under coverage $w_{\mathrm{sur}}(\hat p) = w_{\mathrm{sur}}(\hat q) = 0$, so the bound of
Corollary~\ref{cor:cert} reduces to $\eps_{\mathrm{solve}} + 2\zeta(m,\delta)$ and is at most
$\eps_{\mathrm{tgt}}$ exactly when $2\zeta(m,\delta) \le \eps_{\mathrm{tgt}} -
\eps_{\mathrm{solve}}$. Since $\zeta(m,\delta)$ is strictly decreasing in $m$, the set of
admissible $m$ is an up-set and $m^\ast$ is its minimum. The cost follows because coverage
simulates $B^\ast$ cells and each must carry $m^\ast$ visits. For the rate, drop the
$\log\log_2(2m)$ term, which is $o(\log m)$: solving
$\sigma\sqrt{2\log(4 n_D n_A/\delta)/m} = (\eps_{\mathrm{tgt}} - \eps_{\mathrm{solve}})/2$
gives $m = 8\sigma^2\log(4 n_D n_A/\delta)/(\eps_{\mathrm{tgt}} -
\eps_{\mathrm{solve}})^2$, and reinstating the $\log\log$ term inflates this by a factor
$1 + o(1)$.
\end{proof}

\begin{proof}[Proof of Theorem~\ref{thm:asymptotic}]
For (i), coverage is monotone: a simulated cell is never returned to the surrogate. Let
$\mathcal{R}$ be the set of rows simulated in every column together with the columns
simulated in every row. By the argument of Proposition~\ref{prop:coverage}, whenever the
iteration fails to halt it enlarges $\mathcal{R}$ by at least one element. Since
$|\mathcal{R}| \le n_D + n_A$ and $\mathcal{R}$ never shrinks, at most $n_D + n_A$ such
extensions can occur across the whole sequence, however many times the mixture changes;
after the last, every subsequent solve halts immediately and $w_{\mathrm{sur}} = 0$ holds
permanently.

For (ii), on the stated event Theorem~\ref{thm:eps} gives
\[
\eps \;\le\; \eps_{\mathrm{solve}} + 2\max_{\pi}\bigl[\zeta(m,\delta)
      + \beta\, w_{\mathrm{sur}}(\pi)\bigr],
\]
and $w_{\mathrm{sur}} = 0$ removes the second term. The radius
\[
\zeta(m,\delta) = \sigma\sqrt{\frac{2\bigl(\log(4n_Dn_A/\delta) + \log\log_2(2m)\bigr)}{m}}
\]
is decreasing in $m$ with $\zeta \to 0$: the numerator grows doubly logarithmically while
the denominator grows linearly, so the ratio inside the root tends to $0$ by
L'H\^opital's rule on the $\infty/\infty$ form. The $\log\log$ term therefore costs only a bounded factor at any
practical $m$, and $\zeta$ retains the $m^{-1/2}$ decay of a fixed-design radius; this is
why the budget of Corollary~\ref{cor:budget} inherits the fixed-design rate up to that
factor. Crucially the bound is time-uniform: it holds simultaneously for all $m$ and all
sample-dependent stopping rules on a single event of probability $1-\delta$, so the limit is
taken on that event without a further union bound.

For (iii), $\eps_{\mathrm{solve}}$ is the exploitability of $(\hat p, \hat q)$ in the
estimated game and vanishes when the restricted game is solved exactly, giving $\eps \to 0$.
Since $\eps$ is exploitability in the true restricted game, the limit profile is an exact
equilibrium of it. The cells outside the coverage set enter neither $\max_i (Dq)_i$ nor
$\max_j (Ap)_j$ by Corollary~\ref{cor:coverage}, so their values are unconstrained.
\end{proof}

\begin{proof}[Proof of Lemma~\ref{lem:collision}]
By Eq.~\eqref{eq:fingerprint}, $\varphi(\pi)$ is a function of
$\{\mu_\pi(\xi_p)\}_{p \le P}$ alone, so any two policies agreeing on $\Xi$ receive
identical fingerprints whatever they do elsewhere. Payoffs, by contrast, depend on behaviour
at all reachable states. Choosing $\pi'$ to agree with $\pi$ on $\Xi$ and to realise the
desired behaviour off it, which the hypothesis permits, gives the claim.
\end{proof}

\begin{proof}[Proof of Theorem~\ref{thm:trap}]
Suppose for contradiction that some such $\mathcal{A}$, at budget $B$, simulates a trap cell
on \emph{every} twin--trap instance.

Consider the \emph{all-duplicate} instance $I_0$, in which a designated attacker strategy
$a_\star$ is a genuine copy of $a_{\mathrm{twin}}$: same fingerprint, same payoffs, likewise
weak. Run $\mathcal{A}$ on $I_0$ and let $\mathcal{S}_B$ be the set of cells it simulates, so
$|\mathcal{S}_B| \le B < n_D n_A$. Since $\mathcal{S}_B$ omits at least one cell, and by
enlarging the pool if necessary, choose an attacker column $j_\star$ that $\mathcal{S}_B$
does not meet.

By Lemma~\ref{lem:collision} there is a policy $a_{\mathrm{trap}}$ with
$\varphi(a_{\mathrm{trap}}) = \varphi(a_{\mathrm{twin}})$ whose payoffs make it a dominant
deviation. Let $I_1$ be the instance obtained from $I_0$ by installing $a_{\mathrm{trap}}$
in column $j_\star$, leaving everything else untouched. Then $I_1$ is a twin--trap instance,
and $I_0$ and $I_1$ differ only in the payoffs of cells in column $j_\star$.

We claim $\mathcal{A}$ executes the same run on $I_0$ and $I_1$. We argue by induction over
rounds. In the base case, before any cell is simulated, both runs hold the empty cache and
the same fingerprints, since $\varphi(a_{\mathrm{trap}}) = \varphi(a_{\mathrm{twin}}) =
\varphi(a_\star)$ and $\varphi$ is a function of the actor alone, not of payoffs; the two
runs therefore agree at round $1$. For the inductive hypothesis, suppose the runs agree
through round $r$, so the caches $\mathcal{C}_r$ agree as multisets of simulated cells. The
entries recorded there agree as values as well: every cell in
$\mathcal{C}_r \subseteq \mathcal{S}_B$ lies outside column $j_\star$, where $I_0$ and $I_1$
are identical by construction. Hence $\mathcal{A}$ receives identical inputs at round $r+1$
and, being deterministic, selects identical cells there, so the runs agree through round
$r+1$. By induction, $\mathcal{A}$ selects the same cells on $I_0$ and $I_1$ at every
round.

Therefore $\mathcal{A}$ simulates exactly $\mathcal{S}_B$ on $I_1$, and $\mathcal{S}_B$ does
not meet column $j_\star$, so $\mathcal{A}$ simulates no trap cell on $I_1$. But $I_1$ is a
twin--trap instance, on which the hypothesis asserts that $\mathcal{A}$ does simulate a trap
cell. This is the contradiction, and it establishes the first claim.

For the exploitability claim, on $I_1$ the surrogate assigns column $j_\star$ its prediction
for $a_{\mathrm{twin}}$, since the two share a fingerprint and no label from $j_\star$ ever
enters the cache. That prediction is uniformly weak, so the returned $\hat q$ places no mass
on the trap. In the true game $a_{\mathrm{trap}}$ is a dominant deviation, so the attacker's
regret $\max_j (A\hat p)_j - \hat q^\top A \hat p$ is at least
$\eta \defeq (A\hat p)_{j_\star} - \hat q^\top A \hat p$, the gap between the trap's true
payoff against $\hat p$ and the value the estimated game reports. This is fixed by the
instance and does not shrink with $B$.

Uniform draws escape the argument because the construction needs a deterministic
$\mathcal{S}_B$ outside of which to place the trap. A rule visiting each unsimulated cell
with probability bounded below admits no such placement, and reaches column $j_\star$ almost
surely.
\end{proof}

\subsection{Fingerprint collisions}
\label{app:collision}

\begin{lemma}[Fingerprint collisions exist]
\label{lem:collision}
Fix a probe bank $\Xi = \{\xi_1,\dots,\xi_P\}$ and let $\Pi$ be a policy class in which,
for any $\pi \in \Pi$ and any prescribed behaviour off $\Xi$, some $\pi' \in \Pi$ agrees
with $\pi$ on $\Xi$ and realises that behaviour elsewhere. Then for every $\pi \in \Pi$
there is $\pi' \in \Pi$ with $\varphi(\pi') = \varphi(\pi)$ whose payoffs against a fixed
opponent are arbitrary.
\end{lemma}

The hypothesis on $\Pi$ is mild. It asks only that agreement on $P$ probe states does not
determine behaviour everywhere, which fails only if $\Xi$ happens to be a sufficient
statistic for the policy class. For neural controllers on a state space larger than the
probe bank, it holds.

\subsection{Extended discussion}
\label{app:extended}

This appendix gives the longer treatment of three topics the main text summarises:
how RWPS relates to prior payoff-estimation and equilibrium-search methods, what the
coverage result adds to known EGTA principles and how it prices a certificate, and
further detail on the growing-pool PSRO comparison. Throughout, RWPS is the budgeted estimator of
Algorithm~\ref{alg:build}: it simulates a subset of payoff cells, fills the rest
with an ensemble surrogate over strategy embeddings, and chooses cells by a score
built from bootstrap equilibrium marginals. The \emph{deviation-relevant set} is
the set of cells pairing every strategy of one player with each strategy in the
support of the other player's equilibrium mixture.

\subsubsection{Relation to prior payoff-estimation methods}
\label{app:ext-related}

GS, PSP and PS-REG-M allocate samples so that a
uniform $\eps$-approximation guarantee holds over every cell, pruning profiles
once their bounds separate. We run all three with the sampling schedules from their published works. At the
budgets RWPS targets that guarantee cannot be bought: PSP's published schedule begins at $294$ samples per cell, whereas a $5\%$ per-build
budget on a $14\times14$ pool affords ten cell evaluations in total. None of the
three methods therefore prunes a single profile, and all return the history-only
estimate. This reflects the operating regime, not a flaw in those algorithms,
which are designed for budgets that cover the matrix many times over.

\citet{sokota2019deviation} also pair a
learned model with an acquisition rule defined at the level of the solution. They
train a network mapping mixed profiles to deviation payoffs and refine it by
sampling in the neighbourhood of each candidate equilibrium. Three differences
separate their setting from ours. Their model takes a mixture as input and
returns deviation payoffs directly, so no payoff table is built, whereas RWPS
estimates individual cells over strategy embeddings and passes the completed
matrix to a solver. They exploit role symmetry to scale in the number of players,
while our games are two-player, asymmetric, and have pools that grow by best
response. And their refinement is validated empirically, with no bound on the
resulting equilibrium, which is what Theorem~\ref{thm:eps} and
Corollary~\ref{cor:coverage} supply. Because their rule is a score over cells
rather than a full payoff-build method, the fair comparison runs it inside an
otherwise identical build, which is how Table~\ref{tab:acq} reports it.

Choosing which expensive
evaluation to run is the subject of active learning and Bayesian
optimisation~\citep{settles2009active,shahriari2016bayesopt}, and deep
ensembles~\citep{lakshminarayanan2017ensembles} supply the per-cell uncertainty
RWPS uses. \citet{picheny2016bayesian} apply Bayesian optimisation to Nash
equilibrium computation in black-box games, treating the equilibrium rather than
a payoff as the object of interest. They target continuous strategy spaces with a
Gaussian-process surrogate and a fixed game; RWPS targets a finite pool that grows
across PSRO iterations and reuses every simulated entry as training data. Both
depart from standard Bayesian optimisation in the target: accuracy of the
equilibrium of the whole matrix, not pointwise accuracy. RWPS measures uncertainty
about that equilibrium by re-solving bootstrap realisations of the game, which is
nearly free because solving is cheap next to simulating.

Of their two acquisition criteria, probability of equilibrium has a discrete
analogue, reported as the bootstrap-inclusion variant in
Appendix~\ref{app:acq-variants}: a cell is scored by how often it enters a
bootstrap support instead of by a Gaussian-process posterior. Stepwise
uncertainty reduction does not carry over. It targets the probability that a
cell \emph{is} the equilibrium, which is defined for the pure equilibria of a
discretised space but not for the mixed equilibria, with supports $(2,2)$ and
$(12,12)$, that our games have. Extending it to mixed equilibria is open.

Two lines extend
\citet{picheny2016bayesian}. \citet{aldujaili2018blackbox} search for pure equilibria of
continuous black-box games by Gaussian-process regression on an approximate regret
surface, with an $\eps$-greedy trade-off between exploiting that surface and reducing
uncertainty. \citet{han2024arise} give ARISE, which restricts acquisition to a region of
interest around the current equilibrium estimate and proves no-regret guarantees for the
resulting query sequence. Both minimise a global regret functional over a continuous
strategy space and query strategy profiles directly. RWPS instead completes a finite
payoff matrix whose rows and columns are policies produced by PSRO, reuses every cell
across iterations, and certifies the equilibrium of the completed matrix; ARISE's
region-of-interest restriction plays a role similar to our equilibrium weighting, and
its no-regret analysis is over queries rather than over the returned mixture.

\citet{nguyen2026conservative} also face payoffs that
cannot be freely simulated, but from the opposite direction: their dataset is fixed, so
they replace the simulator with an ensemble dynamics model and steer best responses away
from profiles the data does not support. Our setting keeps the simulator and rations
calls to it, so uncertainty drives \emph{which} cells to pay for rather than a
conservatism penalty on responses. The two are complementary: a budgeted estimator
decides what to simulate next, and a conservative response objective decides what to
trust when nothing more can be simulated.

\citet{rowland2019multiagent}
maintain confidence intervals over payoffs until the response graph is
determined, and \citet{rashid2021alpharank} choose each query to maximise
information gain about $\alpha$-rank. Neither optimises exploitability against a
Nash mixture, so neither is run as a baseline. Information-gain acquisition is
represented among the baselines through the IGS method of
\citet{jordan2008mrfs}.

\subsubsection{Coverage, budgets and certificates}
\label{app:ext-theory}

That verifying a profile's regret
requires only the profile and its unilateral deviations is long established in
empirical game-theoretic analysis (EGTA)~\citep{wellman2006egta}.
\citet{wellman2025egtasurvey} state it directly and note that the count stays
modest for mixed profiles with small supports; the same idea underlies
minimum-regret-first search (MRFS)~\citep{jordan2008mrfs} and
deviation-preserving reduction~\citep{wiedenbeck2012dpr}. We do not claim the
principle. Corollary~\ref{cor:coverage} adds two things: a closed-form count,
$B^\ast = n_D s_A + n_A s_D - s_D s_A$ cells for a bimatrix game with pool sizes
$n_D, n_A$ and support sizes $s_D, s_A$, and a coupling to the error bound of
Theorem~\ref{thm:eps}.

The coupling is the more consequential part. In prior treatments the cells outside
the deviation set are \emph{unevaluated}: pruned, deferred or unknown, with
guarantees stated only over what was sampled. In RWPS they are \emph{filled} by a
surrogate the solver reads, and the corollary shows that however inaccurate that
fill is, it cannot reach either player's regret; the surrogate term
$w_{\mathrm{sur}}$ is zero regardless of the surrogate's error $\beta$. This is
what makes interpolation admissible in place of pruning, and what turns the
deviation-relevant set into a budget for an estimator rather than a stopping rule
for a sampler.

The connection is exact. MRFS confirms a \emph{pure}
profile of a two-player $n\times n$ game by evaluating its entire unilateral
deviation set, the row and column through that cell, which is $2n-1$
evaluations. Setting $s_D = s_A = 1$ in Corollary~\ref{cor:coverage} gives
$n + n - 1 = 2n-1$. The deviation-relevant set is MRFS's confirmation criterion
extended from a pure profile to a mixed one.

We need full columns on the
opponent's support, since only those columns enter a player's expected payoff, and
full rows on a player's own support, since deviations range over all of that
player's strategies. Nothing outside this set affects either player's regret.
Coverage does not simulate the whole matrix: the remaining $n_D n_A - B^\ast$
cells stay surrogate-filled and are read by the solver, and the corollary shows
only that they carry zero weight in regret.

Theorem~\ref{thm:eps} bounds exploitability by a
statistical term, from finite rollouts per simulated cell, and a representational
term $\beta$, from surrogate error on unsimulated cells. Only the statistical term
shrinks with spending. The representational term depends on whether fingerprints
happen to predict payoffs in the game at hand, cannot be certified in advance, and
by Theorem~\ref{thm:trap} can be arbitrarily large. Covering the
deviation-relevant set removes it outright, leaving a purely statistical
certificate: the guarantee a full rebuild would give, without a full rebuild and
without any assumption about surrogate quality.

$B^\ast$ is linear in pool and support sizes, so
support size rather than matrix size sets the price of coverage. The support sizes
are unknown before solving, but they can be estimated by re-solving many
bootstrap realisations of the surrogate-filled game, following the bootstrap
methodology of \citet{wiedenbeck2014bootstrap}; this costs solves, not rollouts.
Before committing any simulation, a practitioner can therefore read off whether
coverage is a small fraction of the matrix, in which case budgeted estimation is
worth enabling, or close to all of it, in which case a full rebuild is the better
choice. Treated as a one-shot prediction the estimate can be wrong;
Appendix~\ref{sec:fixedpoint} shows how to iterate it to a fixed point instead.

Coverage fixes \emph{which} cells are
simulated; the number of rollouts per cell fixes \emph{how tight} the certificate
of Corollary~\ref{cor:cert} is.

\begin{corollary}[Budget for a target certificate]
\label{cor:budget}
Fix $\delta$ and a target $\eps_{\mathrm{tgt}} > \eps_{\mathrm{solve}}$. If the
deviation-relevant set is covered, so that $w_{\mathrm{sur}} = 0$, then the certificate of
Corollary~\ref{cor:cert} is at most $\eps_{\mathrm{tgt}}$ as soon as every simulated cell
carries
\[
m \;\ge\; m^\ast \;=\; \min\{\, m : 2\zeta(m,\delta) \le \eps_{\mathrm{tgt}} - \eps_{\mathrm{solve}} \,\},
\]
and the total rollout cost of achieving it is $B^\ast \cdot m^\ast$, with
\[
m^\ast \;=\; O\!\left(\frac{\sigma^2 \log(n_D n_A/\delta)}{(\eps_{\mathrm{tgt}} - \eps_{\mathrm{solve}})^2}\right)
\]
up to the $\log\log$ factor carried by the time-uniform radius.
\end{corollary}

Here $\zeta(m,\delta)$ is the concentration radius for a cell mean from $m$
rollouts, $\sigma$ the payoff noise scale, and $\eps_{\mathrm{solve}}$ the solver's
own regret. The two factors are bought separately and neither substitutes for the
other: $B^\ast$ removes $\beta$, and $m^\ast$ shrinks $\zeta$. Both are computable in
advance, $B^\ast$ from bootstrap supports and $m^\ast$ by inverting $\zeta$, and
every quantity in the certificate is computable from simulation data. A deployment
can therefore quote a budget against a target exploitability before spending it,
including where the true equilibrium is unavailable. Coverage is cheap and
precision is not: at $\sigma = 0.1$ on a $30\times30$ pool, $2\zeta$ is roughly half
the payoff range at $m = 4$, and reaching $0.05$ needs $m^\ast \approx 512$, while
pool size enters only logarithmically.

The unsigned forms of Eq.~\eqref{eq:chain} charge errors in both directions, but the signed
form (a) does not: only a deviation row whose value the estimate understates, or a support
whose value it overstates, produces regret. This is what makes optimism admissible. Filling
the unsimulated cells of a player's own matrix at the maximum payoff, in the tradition of
R-max~\citep{brafman2002rmax}, understates no deviation, so the deviation term reduces to
statistical noise on simulated cells and only the support term remains; Corollary~\ref{cor:cert}
uses exactly this fill, and only to evaluate the bound after the solve. Using the same fill
inside the solve instead changes the object solved. Every unsimulated cell then looks
maximal, so a best response points at it; it enters the support, is simulated, proves
ordinary, and leaves, and the next unsimulated cell takes its place. With $n_D n_A$ cells to
chase, this certifies only after nearly the whole matrix is bought
(Appendix~\ref{app:diagnostics}). Optimism certifies a solution; it does not save
simulation.

\subsubsection{Necessity of forced exploration}
\label{sec:twintrap}

The asymptotic guarantee of RWPS rests on its forced-exploration draws. They
visit every cell infinitely often almost surely, so running means converge and
Lemma~\ref{lem:supp} with $\tsupp \to 0$ closes the argument; the two-phase
Borel--Cantelli proof is standard. No acquisition rule that reads only
fingerprints and cached payoffs can replace these draws. A fingerprint is a finite
summary of a policy's behaviour, so distinct policies can share one while
differing arbitrarily in payoff (Appendix~\ref{app:collision}), and such a rule
cannot tell them apart.

\begin{theorem}[Necessity of exploration]
\label{thm:trap}
Call an attacker strategy $a_{\mathrm{twin}}$ a \emph{twin} if it is densely simulated and
uniformly weak, and $a_{\mathrm{trap}}$ a \emph{trap} if
$\varphi(a_{\mathrm{trap}}) = \varphi(a_{\mathrm{twin}})$ while its true payoffs make it a
dominant deviation. Let $\mathcal{A}$ be any deterministic acquisition rule whose selections
are a function of the fingerprints $\{\varphi(s)\}_{s \in S_D \cup S_A}$ and the cache
$\mathcal{C}$ alone. Then for every budget $B < n_D n_A$ there is a twin--trap instance on
which $\mathcal{A}$ simulates no trap cell, and on which the returned mixture has
exploitability at least $\eta > 0$ with $\eta$ determined by the instance and independent
of $B$.
\end{theorem}

\subsubsection{Progressive sampling at budgets where it prunes}
\label{app:ps-high}

The growing-pool comparison runs the progressive-sampling family below the budgets it was
designed for. To show where it does become competitive, Table~\ref{tab:psp} in the main paper, and Table~\ref{tab:psp-full} in full, compare RWPS
with GS, PSP and PS-REG-M on fixed $21\times21$ pools at budgets from $0.2\times$ to
$100\times$ the $441$-cell matrix, with each family member implemented from its original
specification. PS-REG-M, the member with mixed-equilibrium guarantees, prunes nothing below
$100\times$; there it prunes $0.5\%$ of utility indices on the informative game and $6.1\%$
on Blotto. PSP prunes nothing at any budget tested, and the pure-equilibrium variant
PS-REG+ starts earlier, pruning $0.7\%$ and $4.2\%$ at $20\times$. At $100\times$ PSP and
PS-REG-M reach exploitability comparable to RWPS: none of the four differences is
significant (paired $t$-test, $p \ge 0.15$). Below that, GS and PSP coincide exactly at
$0.2\times$, because neither can certify any cell and both fall back to a flat fill. Across
all thirty pairwise comparisons RWPS is significantly better in twenty and never
significantly worse ($p<0.05$). The comparison is therefore one of operating regime: the
family buys a uniform PAC guarantee and pays its sample complexity, which is affordable
only at budgets many times the matrix. This sweep predates the final acquisition rule and
uses four seeds, so we report it as a regime check rather than as a head-to-head result.

\begin{table}[h]
\centering\small
\setlength{\tabcolsep}{3pt}
\caption{All five budgets of the progressive-sampling comparison (extends
Table~\ref{tab:psp}); exploitability mean $\pm$ SD over four seeds, bold marks each row's best.}
\label{tab:psp-full}
\begin{tabular}{@{}llccccc@{}}
\toprule
Game & Budget & RWPS & GS & PSP & PS-REG-M & pruned \\
\midrule
\multirow{5}{*}{Inf.}
 & $0.2\times$ & $\mathbf{.051 \pm .013}$ & $.307 {\scriptstyle\pm} .112$ & $.307 {\scriptstyle\pm} .112$ & $.166 {\scriptstyle\pm} .136$ & $0$ \\
 & $1\times$   & $\mathbf{.029 \pm .023}$ & $.162 {\scriptstyle\pm} .138$ & $.222 {\scriptstyle\pm} .081$ & $.214 {\scriptstyle\pm} .106$ & $0$ \\
 & $5\times$   & $\mathbf{.037 \pm .019}$ & $.230 {\scriptstyle\pm} .091$ & $.354 {\scriptstyle\pm} .195$ & $.546 {\scriptstyle\pm} .331$ & $0$ \\
 & $20\times$  & $\mathbf{.012 \pm .006}$ & $.413 {\scriptstyle\pm} .125$ & $.156 {\scriptstyle\pm} .298$ & $.029 {\scriptstyle\pm} .050$ & $0$ \\
 & $100\times$ & $.007 {\scriptstyle\pm} .003$ & $.025 {\scriptstyle\pm} .040$ & $.005 {\scriptstyle\pm} .003$ & $\mathbf{.003 \pm .002}$ & $0.5\%$ \\
\midrule
\multirow{5}{*}{Blotto}
 & $0.2\times$ & $\mathbf{.135 \pm .045}$ & $.488 {\scriptstyle\pm} .024$ & $.488 {\scriptstyle\pm} .024$ & $.501 {\scriptstyle\pm} .041$ & $0$ \\
 & $1\times$   & $\mathbf{.043 \pm .011}$ & $.475 {\scriptstyle\pm} .029$ & $.471 {\scriptstyle\pm} .026$ & $.617 {\scriptstyle\pm} .268$ & $0$ \\
 & $5\times$   & $\mathbf{.032 \pm .018}$ & $.556 {\scriptstyle\pm} .093$ & $.419 {\scriptstyle\pm} .050$ & $.440 {\scriptstyle\pm} .080$ & $0$ \\
 & $20\times$  & $\mathbf{.025 \pm .004}$ & $.561 {\scriptstyle\pm} .081$ & $.366 {\scriptstyle\pm} .081$ & $.404 {\scriptstyle\pm} .049$ & $0$ \\
 & $100\times$ & $.015 {\scriptstyle\pm} .003$ & $.440 {\scriptstyle\pm} .035$ & $\mathbf{.013 \pm .001}$ & $.017 {\scriptstyle\pm} .007$ & $6.1\%$ \\
\bottomrule
\end{tabular}
\end{table}

\subsubsection{Growing-pool PSRO: further detail}
\label{app:ext-results}

In the growing-pool comparison of
Table~\ref{tab:psro}, every method starts from the same pool, but each method's
best-response oracle trains against that method's own restricted equilibrium, so a
better payoff estimate steers PSRO toward different strategies and the pools
diverge after the first iteration. Exploitability is measured against the full
game by zero-lifting each restricted mixture, which keeps pools of different
composition comparable. The comparison is therefore end-to-end: it credits RWPS
both for its payoff estimates and for the strategies those estimates lead PSRO to
find. Appendix~\ref{app:embedding} repeats it with identical payoffs and
uninformative fingerprints to isolate what the fingerprint contributes.

RWPS starts behind: at iteration $0$ its
exploitability is $0.262$, against $0.180$ for MRFS and $0.208$ for uniform
sampling. This is expected. The cache is empty at the first build, so the
surrogate is a rank-one fit to almost no labels and there is no history to reuse.
The advantage of RWPS lies in what it carries between builds, and its lead widens
monotonically with the iteration index.

IGS finishes behind both MRFS and RWPS on both games (Table~\ref{tab:psro}). \citet{jordan2008mrfs} recommend IGS
for noisy payoffs, where MRFS is their method for exactly revealed payoffs, and
report it outperforming the ECVI benchmark. We implement it from their Section~6.2:
the next profile sampled is the one maximising Kullback--Leibler information gain
aggregated over its deviation set, with the minimum-regret probability of each
profile taken from their point approximation, and an MRFS prefix supplying initial
samples as in their IGS-MRFS-3 variant.

The budget, not the algorithm, is what puts IGS behind. IGS reallocates samples
among profiles it has already visited; its original experiments run at roughly
twenty-eight samples per profile on a thirty-five-profile game. A budget that
cannot visit every profile once denies it the quantity it reasons about. On a
fixed $21\times21$ pool, IGS improves from $0.177$ to $0.041$ as its allowance
grows from $0.2$ to $1.0$ evaluations per profile, converging toward MRFS.

Figure~\ref{fig:nashconv_simulator_calls}
reports an earlier CyGym comparison of RWPS and uniform sampling as a function of
cumulative simulator calls.

\begin{figure}[h]
    \centering
    \includegraphics[width=0.5\textwidth]{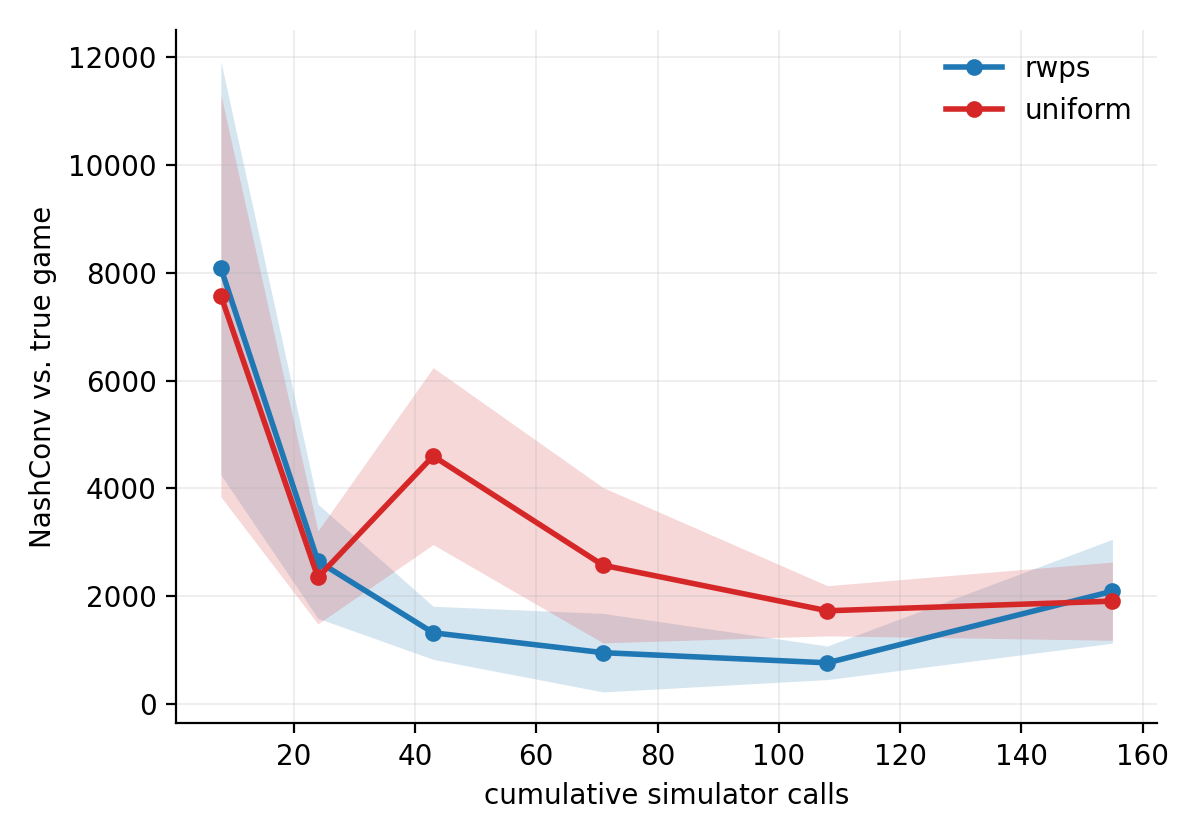}
    \caption{NashConv against cumulative simulator calls in CyGym (lower is better).
    At the lowest and highest budgets uniform sampling suffices; across all budget
    levels RWPS meets or exceeds it.}
    \label{fig:nashconv_simulator_calls}
\end{figure}

The bound evaluated in
Figure~\ref{fig:perturb} includes the solver's own regret $\eps_{\mathrm{solve}}$.
Without that term the signed bound is violated on three runs, all on Blotto and
each by less than that run's solver regret. Fictitious play converges slowly on
games with large supports, and raising its iteration count from $400$ to $20{,}000$
does not remove the residual: measured $\eps_{\mathrm{solve}}$ averages $0.008$ on
Blotto against $0.002$ on the informative game. Because $\eps_{\mathrm{solve}}$ is
computable from the estimated matrices alone, including it costs nothing in
applicability and restores validity on every run, at a cost of $0.006$ in mean
tightness.

\subsubsection{Cyber-defence simulators}
\label{app:simulators}

RWPS needs a simulator in which both an attacker and a defender are learning players, since
PSRO grows a policy pool for each side and every payoff entry pairs one policy from each. We
surveyed the environments in current use against that requirement; Table~\ref{tab:sims}
summarises the result.

CybORG~\citep{standen2021cyborg}, which underpins the CAGE challenges~\citep{cage2023}, and
CyGym~\citep{lanier2026cygym} both support learning agents on each side, and we use the
latter directly. Cyberwheel~\citep{oesch2024cyberwheel} supports one learning red and one
learning blue agent simultaneously and is the third simulator in our evaluation.
gym-idsgame~\citep{hammar2020finding} models intrusion prevention as a Markov game between
a learned attacker and a learned defender, and its accompanying study trains both by
self-play against an opponent pool. It is the closest prior use of a policy pool in this
domain, though it estimates no payoff matrix and so has no counterpart to our estimator. The
multiagent extension of CyberBattleSim~\citep{kunz2022multiagent} adds a trainable defender
and reports joint training of red and blue policies, which makes it a candidate once its
code is publicly maintained.

Several widely used environments do not meet the requirement. In
YAWNING-TITAN~\citep{andrew2022yawningtitan} only the defender learns and the attacker is
scripted and probabilistic, and in the standard CyberBattleSim~\citep{cyberbattlesim2021} the
defender is a predefined stochastic process. NASim~\citep{schwartz2019autonomous} models a
penetration-testing attacker with no active defender. CSLE~\citep{hammar2026csle} is a broad
platform combining emulation and simulation across several security-management use cases
rather than a single two-player environment; we have not identified a configuration of it
that fits our setting.

\begin{table}[h]
\centering\small
\caption{Cyber-defence simulators by whether each side is a learning player.
NASim has no defender at all.}
\label{tab:sims}
\begin{tabular}{@{}lccl@{}}
\toprule
Simulator & Learning red & Learning blue & Role here \\
\midrule
CyGym                     & \cmark & \cmark & evaluated \\
Cyberwheel                & \cmark & \cmark & evaluated \\
CybORG / CAGE             & \cmark & \cmark & related \\
gym-idsgame               & \cmark & \cmark & related \\
Multiagent CyberBattleSim & \cmark & \cmark & candidate \\
YAWNING-TITAN             & \xmark & \cmark & excluded \\
CyberBattleSim            & \cmark & \xmark & excluded \\
NASim                     & \cmark & \xmark & excluded \\
CSLE                      & \multicolumn{2}{c}{platform} & deferred \\
\bottomrule
\end{tabular}
\end{table}

\subsection{The Agentic Network Security Game}
\label{app:ansg}

This appendix gives the full specification of the ANSG introduced in
Section~\ref{sec:experiments}.

CyGym~\citep{lanier2026cygym} models compromise as binary node access: the
attacker exploits a node and thereby controls it. In a network of LLM agents
the relevant failure is different. An agent processes workloads by invoking
tools, and an adversary who plants instructions in content the agent reads can
cause it to execute attacker-chosen goals \emph{in addition to} the ones its
principal asked for. The legitimate task still completes, so the subversion
need not look like a failure at all. ANSG replaces access-compromise with this
\emph{goal subversion} model, keeping CyGym's binary per-node state and its
PSRO solution machinery while changing what compromise means and what each
player can observe.

ANSG is a two-player general-sum partially observable stochastic game
\[
\mathcal{G} = \langle \mathcal{S},\, O^D,\, O^A,\, \mathcal{A}^D,\,
\mathcal{A}^A,\, \mathcal{T},\, U^D,\, U^A,\, \gamma \rangle
\]
over a directed communication graph $G = (\mathcal{V},\mathcal{E})$ whose nodes
host agents. Each node carries a behavioural mode
$b_v \in \{\mathtt{benign},\mathtt{malicious}\}$ and a context age $h_v$. Both
players choose policies simultaneously without observing the other's, so Nash
equilibrium is the solution concept and PSRO the approximation method, as in
CyGym.

Attacker and defender valuations are independent rather than opposed: the
attacker gains the value of the subverted goals it executes, while the defender
loses the value of its workloads being corrupted, and these need not sum to a
constant. Per-goal value factors differ by player. This is the property that
makes ANSG a natural test for a budgeted estimator whose theory assumes nothing
about zero-sum structure.

The defender may \emph{audit} a node, receiving a noisy signal whose
detectability rises with the value of the subverted goal; \emph{reprompt} it, a
soft intervention with no CyGym analogue; \emph{reset} it, which restores
benign behaviour but erases context history at a cost; or \emph{isolate} it,
severing communication edges and halting propagation. The attacker may
\emph{inject} a payload into a workload, choosing a target, a goal rank, and a
channel, or \emph{probe} to estimate which previously compromised nodes remain
malicious.

Injection carries a channel choice $\chi$. Indirect prompt injection appends an
attacker goal alongside the legitimate workload, which therefore completes
normally; the only signal is a rank-dependent trace. Malicious tool triggering
instead redirects an existing call, denying the legitimate output for a more
damaging hit at higher detection probability. The rank choice sets value
against exposure, and the channel choice compounds it.

Both players are partially informed, which is the main departure from CyGym. The defender sees only noisy audit signals. The attacker cannot observe
resets, and because a reset benign node is behaviourally indistinguishable from
one never compromised, must probe to estimate the current state. Compromise
also propagates: a malicious agent can embed payloads in its outputs, infecting
downstream agents along $\mathcal{E}$.

ANSG is the harder of the two deployment games for a budgeted build, for
reasons the theory anticipates. Payoffs require full rollouts of a partially
observable game, so simulation dominates iteration cost and the premise of
Section~\ref{sec:results} holds strongly. Equilibrium supports are not known in
advance, so $B^\ast$ must be estimated from bootstrap replicates rather than
read off a solved game. And no ground-truth equilibrium is computable, so the
certificate of Corollary~\ref{cor:cert} is the only available quality measure
rather than a diagnostic checked against a known answer.
\subsection{Budget sweep on Cyberwheel}
\label{app:cyberwheel}

When the pool is flat, every arm that buys cells returns essentially the same
equilibrium, and budgeted acquisition has nothing to separate.
Figure~\ref{fig:cyberwheel-sweep} repeats the budget sweep of
Section~\ref{sec:cyber-sweep} on Cyberwheel~\citep{oesch2024cyberwheel}. Uniform sampling, IGS,
MRFS, and both RWPS variants all sit within about $0.015$ of zero exploitability
at every budget, with overlapping bands from budget~2 onward. The one visible
separation is at budget~1, where RWPS (uniform) is highest of the buying arms.
GS/PSP/PS-REG-M is constant near $0.16$ because, as on CyGym and ANSG, its
sampling schedule cannot afford a single cell and it returns the history-only
estimate. The pool is flat in a measurable sense: on this $6\times6$ pool the
payoff signal-to-noise ratio is below $1$ for both roles, so no cell's
simulated value is distinguishable from rollout noise.

This is the outcome Corollary~\ref{cor:coverage} predicts for a matrix on which
no cell is much more deviation-relevant than another: if the surrogate is
already accurate everywhere, which cells are bought barely changes the
equilibrium. We report it as the regime in which RWPS offers no advantage over
simple baselines, and in which it also costs nothing.

\begin{figure}[h]
\centering
\includegraphics[width=0.72\linewidth]{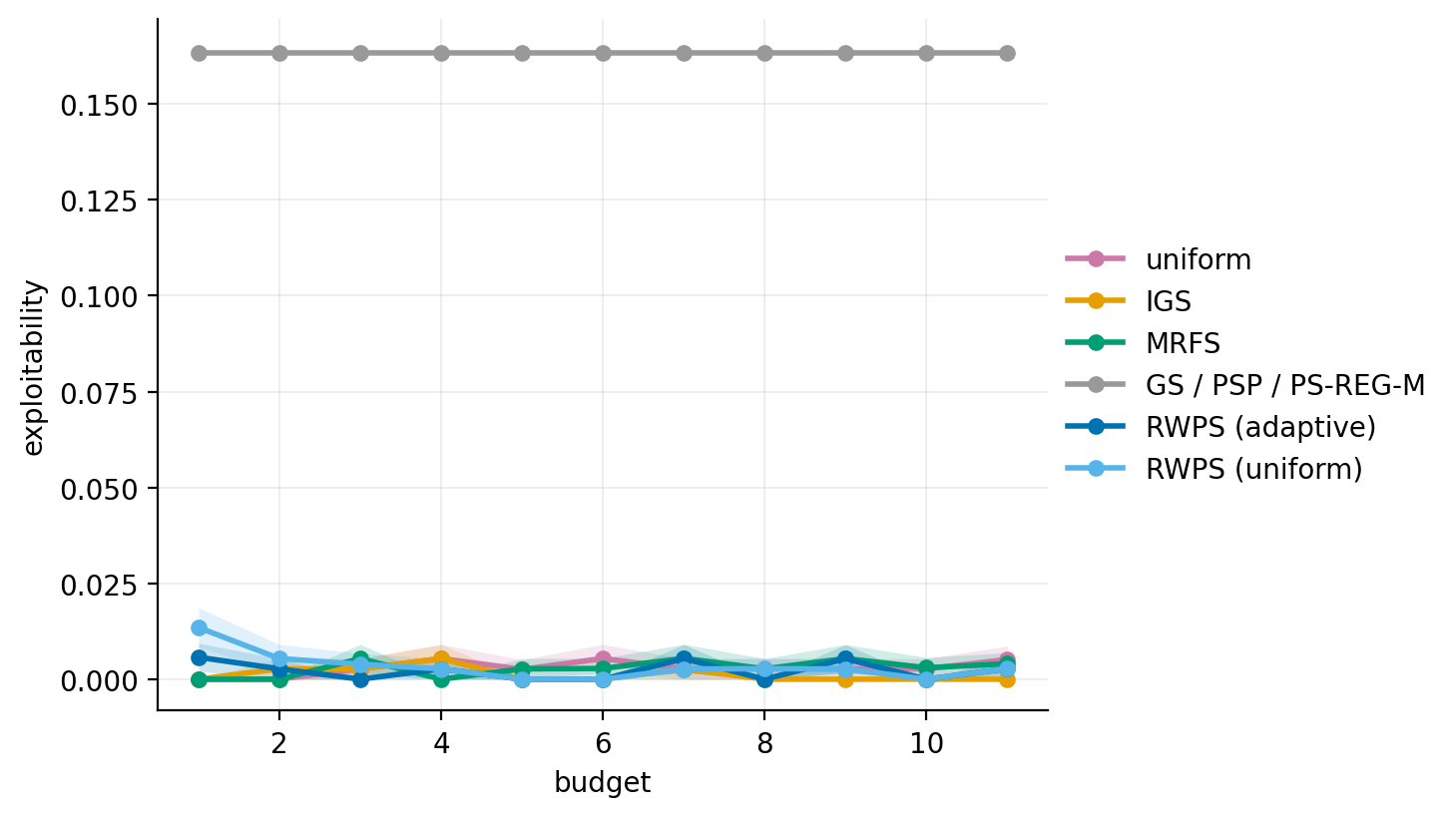}
\caption{Cyberwheel budget sweep. $6\times6$ pool with a $5\times5$ history
pre-cached for every arm, leaving $11$ cells to buy; $16$ seeds. Exploitability is
computed against the stored payoff matrix of the pool. Bands are $\pm1$ standard
error. All buying arms coincide within their bands;
GS/PSP/PS-REG-M buys nothing and returns the history-only estimate.}
\label{fig:cyberwheel-sweep}
\end{figure}

\subsection{Implementation and computational procedure}
\label{app:procedure}

A build with budget $B$ runs $R$ rounds. Each round refits the ensemble on the whole cache,
predicts $\hat u_{ij}$ and $\hat\sigma_{ij}$ on every unsimulated cell, re-solves
$B_{\mathrm{boot}}$ bootstrap realisations of the game to obtain the smoothed marginals, and
then simulates about $B/R$ cells, each uniform with probability $\eps_{\mathrm{x}}$ and
otherwise the highest-scoring. The surrogate is therefore refit $R+1$ times per build, once
per round and once more before the matrices are committed. Only the simulation step
consumes budget; the bootstrap costs solves alone.

The cache persists across PSRO iterations, so a fit at iteration $k$ trains on every cell
bought earlier, and cached values always override the surrogate. Repeat visits, which
arise only from exploration draws, are averaged into a running mean. The opening fit of
each build warm-starts from the previous build's weights, which leaves final
exploitability unchanged ($p = 0.81$ informative, $p = 0.62$ Blotto, eight seeds) and cuts
build time by $17$--$20\%$; the later rounds of a build are fit cold, because the cache
roughly doubles between rounds and a short warm schedule does not track that. Until the
cache holds eight labels, a rank-one row-and-column model stands in for the ensemble, so
the first round on a fresh pool is close to uninformed.

Each ensemble member is an MLP mapping the concatenated embeddings of a cell's two
strategies to both players' payoffs. Member $m$ trains on its own bootstrap resample of the
cache ($n$ cells drawn with replacement) with its own seed, which is the only source of
disagreement between members; the loss is mean squared error on standardised targets,
optimised full-batch with Adam for a fixed number of epochs, with no validation split or
early stopping. This ensemble bootstrap measures uncertainty about individual payoffs. The
separate acquisition bootstrap resamples $80\%$ of each pool's strategies and re-solves the
game, and so measures uncertainty about the equilibrium.

On the synthetic games the restricted game is solved by fictitious play, whose residual
regret is carried in every bound as $\eps_{\mathrm{solve}}$; the CyGym, ANSG and Cyberwheel
sweeps use CyGym's solver, support enumeration with a Lemke--Howson fallback. Those sweeps
also embed strategies by identity tags, one-hot vectors per strategy, since building a
behavioural probe bank for each simulator was out of scope; with tags the surrogate learns
only row and column effects. On ANSG, $36$ cached cells were enough to trigger the MLP,
which then filled the $13$ unseen cells worse than a flat mean, so that sweep keeps every
RWPS arm on the rank-one model throughout. Table~\ref{tab:hparams} lists the settings.

\begin{table}[h]
\centering\small
\caption{Default settings unless stated otherwise; epochs are cold / warm fits.}
\label{tab:hparams}
\begin{tabular}{@{}ll@{}}
\toprule
Setting & Value \\
\midrule
Ensemble                    & $M = 3$ MLPs, $2\times48$ hidden \\
Training                    & Adam, lr $5\times10^{-3}$; $60$ / $10$ epochs \\
Rounds and batch            & $R = 3$, $c = \lceil B/3\rceil$ \\
Bootstrap                   & $B_{\mathrm{boot}} = 32$, $80\%$ of each pool \\
Smoothing, exploration      & $\nu = 0.25$, $\eps_{\mathrm{x}} = 0.15$ \\
Rank-one fallback           & below $8$ cached labels \\
Fictitious play             & $400$ final, $200$ in bootstrap \\
Rollouts                    & $N_{\mathrm{MC}} = 4$, $\sigma = 0.10$ \\
Synthetic games             & $21\times21$, payoffs in $[0,1]$ \\
Probe bank                  & $P = 32$ states per role \\
\bottomrule
\end{tabular}
\end{table}

\subsection{Embedding Analysis}
\label{app:embedding}

The latent-quality game admits a controlled ablation of the fingerprint. Holding
the payoff matrices fixed and redrawing the fingerprints independently of latent
quality yields a game identical in every respect except that the embedding
carries no payoff information. Table~\ref{tab:embedding} repeats the
growing-pool comparison of Section~\ref{sec:psro-headline} on both versions.

\begin{table}[t]
\centering\small
\caption{Growing-pool PSRO on games with identical payoffs and different
fingerprints; final exploitability, mean over five seeds, all budgeted methods at a
budget of $256$. Only the estimator reads fingerprints, so only its row moves.}
\label{tab:embedding}
\begin{tabular}{@{}lccc@{}}
\toprule
Payoff build & informative & uninformative & $\Delta$ \\
\midrule
RWPS (ours)          & $0.045$ & $0.060$ & $+0.015$ \\
MRFS                 & $0.053$ & $0.053$ & $0$ \\
IGS                  & $0.082$ & $0.082$ & $0$ \\
GS / PSP / PS-REG-M  & $0.269$ & $0.269$ & $0$ \\
\emph{full, unbudgeted} & $0.035$ & $0.035$ & $0$ \\
\bottomrule
\end{tabular}
\end{table}

Every method that does not consult the embedding is unchanged to the reported
precision, which is a check on the construction as much as a result. The
estimator is the only row that moves, so $0.015$ is a direct measurement of what
payoff-informative fingerprints are worth in this game.

The estimator's advantage over
uniform sampling splits into a part that survives an uninformative embedding,
$0.092 \rightarrow 0.060$, and a part attributable to generalisation across
embedding space, $0.060 \rightarrow 0.045$. The first is the larger. With
fingerprints that are pure noise the surrogate still learns row and column
effects from cells already simulated: it learns that a given strategy tends to
score well or badly against the pool, which requires no embedding at all, and that
alone recovers most of the gap. This is why the
rank-one row-and-column fallback of Appendix~\ref{app:procedure} is not merely a
cold-start device: it is a substantial fraction of what the surrogate does
throughout.

It also explains where the estimator stops winning. On the uninformative game
MRFS reaches $0.053$ against the estimator's $0.060$, inside one standard
deviation. MRFS obtains the same row-and-column structure directly, by expanding
a profile into its unilateral deviations, and needs no embedding to do it. When
the embedding contributes nothing there is nothing left for the estimator to win
with, and the two methods are measuring the same thing by different means. We
report this as the intended outcome of the ablation rather than as a defeat: the
uninformative game exists to isolate the embedding term, and a near-tie with a
model-free method is the correct result when that term is zero by construction.

\subsection{Coverage iteration and asymptotic guarantee}
\label{app:coverage}
\label{sec:fixedpoint}

Corollary~\ref{cor:coverage} identifies the deviation-relevant set $B^\ast$ and
Corollary~\ref{cor:budget} prices a target certificate against it. This appendix
turns the pair into an iteration and gives the asymptotic guarantee that follows.
The finite-budget result the paper relies on is Theorem~\ref{thm:eps}, which does
not depend on anything here.

Corollary~\ref{cor:cert} reports what a completed build certifies; run backwards, it answers
how much simulation buys a stated guarantee, before any of it is spent.

We do not treat $B^\ast$ as a prediction that has to be got right. Instead we treat it as the
first iterate of a procedure that comes with a free stopping test. The test is
$w_{\mathrm{sur}} = 0$, which holds exactly when the deviation-relevant set of the returned
$(\hat p, \hat q)$ has been covered, and which is a cache-membership query costing no
rollouts. This dissolves the apparent circularity in reading $B^\ast$ as an a-priori budget.
The supports it needs are a property of the solve that it funds, but the procedure does not
need them to be correct, only to be verified at the end.

\begin{proposition}[Coverage iteration]
\label{prop:coverage}
From any cache, iterate: solve $(\Dhat_t, \Ahat_t)$ for $(\hat p_t, \hat q_t)$; simulate
every uncached cell of
$\{1{:}n_D\}\times\mathrm{supp}(\hat q_t) \cup \mathrm{supp}(\hat p_t)\times\{1{:}n_A\}$;
refit and repeat, halting when
$w_{\mathrm{sur}}(\hat p_t) = w_{\mathrm{sur}}(\hat q_t) = 0$. Then
\begin{enumerate}[label=(\roman*),leftmargin=*,topsep=2pt,itemsep=1pt]
\item the procedure halts after at most $n_D + n_A$ rounds;
\item at halting the hypothesis of Corollary~\ref{cor:coverage} holds for the equilibrium
      actually returned, so $\eps \le \eps_{\mathrm{solve}} + 2\zeta(m,\delta)$ with no
      dependence on $\beta$;
\item total simulation never exceeds $n_D n_A$ cells, so the procedure is never more
      expensive than the full rebuild it replaces.
\end{enumerate}
\end{proposition}

Halting is by exhaustion of rows and columns, so the sequence of mixtures need not converge,
and an oscillating equilibrium costs additional rounds and nothing else. The procedure is
sound and terminating but not minimal, since an intermediate mixture can pull in rows for a
support that later vacates. Note that the two error directions are not symmetric.
Overshooting the support buys unnecessary cells while keeping the guarantee, whereas
undershooting forfeits it, so we seed the first iterate with the \emph{union} of the
bootstrap supports rather than with a threshold on their average.

\begin{theorem}[Convergence under coverage]
\label{thm:asymptotic}
Consider a sequence of builds in which the iteration of Proposition~\ref{prop:coverage} is
run to halting whenever the returned mixture changes, and in which every cell of the
coverage set continues to be sampled. Then, on the event of probability at least $1-\delta$
on which Theorem~\ref{thm:eps} holds,
\begin{enumerate}[label=(\roman*),leftmargin=*,topsep=2pt,itemsep=1pt]
\item at most $n_D + n_A$ coverage extensions occur over the entire sequence, after which
      $w_{\mathrm{sur}}(\hat p) = w_{\mathrm{sur}}(\hat q) = 0$ holds permanently;
\item thereafter $\eps \le \eps_{\mathrm{solve}} + 2\zeta(m,\delta)$, and
      $\eps \to \eps_{\mathrm{solve}}$ as $m \to \infty$;
\item hence, with the restricted game solved exactly, the returned mixture converges to an
      exact equilibrium of that game, while the $n_D n_A - B^\ast$ surrogate-filled cells
      are never corrected and may remain arbitrarily wrong.
\end{enumerate}
\end{theorem}

A single confidence event supports every $m$ here, which is what the time-uniform radius
buys us: a fixed-design bound holds at one prescribed sample size, so letting
$m \to \infty$ under it would require a union over sample sizes. Neither half of the result
is an assumption. We reach coverage in finitely many rounds using a stopping test that is a
cache query, and the vanishing of $\zeta$ follows from continued sampling of the cells that
coverage identifies. Coverage decides which cells convergence is possible on and visits
decide how fast, so spending alone does not substitute for coverage. The contrast with
Section~\ref{sec:twintrap} is one of cost: that argument needs forced exploration to visit
every cell infinitely often, $\Theta(n_D n_A)$ cells, whereas this one needs only the
$\Theta(n_D s_A + n_A s_D)$ cells of the coverage set.

Coverage guarantees the bound, not the solve. Driving $w_{\mathrm{sur}}$ to zero
removes $\beta$ from Theorem~\ref{thm:eps} by construction, so the bound holds
whatever the surrogate predicts. The certificate no longer depends on surrogate
quality. The surrogate has not ceased to matter. Which equilibrium the solver
arrives at is still a function of the fill, because the solver reads the complete
matrix and the unsimulated cells are exactly the cells that determined they were
unplayed. The surrogate therefore controls rounds to termination rather than the
validity of what the iteration returns. The $\beta \cdot w_{\mathrm{sur}}$ term
cannot express this role, since coverage sets it to zero. At $B^\ast$ coverage a
surrogate fill and a flat fill carry the same guarantee, and
Section~\ref{sec:results} separates them empirically.

Coverage separates the budget into two parts that are priced independently.
Coverage decides which cells are bought and removes $\beta \cdot
w_{\mathrm{sur}}$. Revisits decide how tight the bound is and shrink
$2\zeta(m,\delta)$. Only the first is a question about the matrix, and it is the
cheap one, costing $s_D + s_A$ cells per iteration under stable supports. The
second is a visit count, and at small $m$ it dominates, so a covered build
certifies very little until $m$ grows. A deployment can invert this: fix the
exploitability bound required, take $m$ to be the smallest visit count with
$2\zeta(m,\delta)$ below it, and the rollout budget follows as the coverage set
times $m$. Section~\ref{sec:results} prices this.

The iteration trades a spending cap for a guarantee. A budget fixed as a fraction
of the matrix bounds what a build can cost. The iteration instead spends what the
game demands, which is $s_D + s_A$ cells per iteration under stable supports and,
in the worst case, the full matrix. Under a hard rollout ceiling it may be
truncated before it halts. Nothing is lost when that happens. Stopping with
$w_{\mathrm{sur}} > 0$ returns us to Theorem~\ref{thm:eps} with the
representational term retained at the measured $w_{\mathrm{sur}}$, which is the
bound the estimator carries in any case. Halting is what removes $\beta$;
truncation degrades to the general bound rather than forfeiting it.

Fingerprint collisions do not defeat coverage. Theorem~\ref{thm:trap} constrains
acquisition rules that read fingerprints, and its proof plants the trap in a
column the rule never touches. Coverage admits no such column:
$w_{\mathrm{sur}}(\hat p) = 0$ requires $\mathrm{supp}(\hat p) \times \{1{:}n_A\}$, so every
attacker strategy is simulated against the defender's support whatever its fingerprint
reports, and the construction has nowhere to place the trap. If a trap's true payoffs make
it a profitable deviation against $\hat p$, covering reveals them, the re-solve moves mass
onto it, $w_{\mathrm{sur}}$ becomes positive, and the next round completes its column. If it
is profitable only against strategies outside $\mathrm{supp}(\hat p)$, then it does not
enter $\max_j (A\hat p)_j$ and is not a deviation by which the returned mixture can be
exploited. Forced exploration remains load-bearing for the acquisition rule of
Section~\ref{sec:acquisition}, whose score is a function of the embedding, but it is not
needed for the coverage guarantee.

Coverage costs $s_D + s_A$ cells per PSRO iteration whenever a new strategy pair
leaves the supports unchanged. With a persistent cache the right baseline is not
a $\Theta(n_D n_A)$ rebuild. It is the
$2k{+}1$ cells that a pool of size $k$ per side newly exposes, an incremental-caching
accounting standard in EGTA~\citep{wellman2006egta,wiedenbeck2012dpr}. The result
is immediate: those are the only cells of the deviation-relevant set the cache
does not already hold, which is to say the only cells at which
$w_{\mathrm{sur}} \neq 0$. The condition is support stability rather than support
size, and it is the regime PSRO enters near termination. The saving accrues late
in a run rather than eroding.

Each result above attaches to a different part of the estimator, and none of them derives
the acquisition rule. The surrogate fill steers \emph{which} equilibrium is solved for.
Coverage makes the \emph{bound} on that equilibrium independent of the fill
(Proposition~\ref{prop:coverage}). Corollary~\ref{cor:cert} \emph{certifies} the result
post hoc in settings where no ground truth exists, and Corollary~\ref{cor:budget} inverts
that certificate into a budget.

One caveat is worth stating. Our acquisition rule targets value sensitivity, through $W$ on
support$\times$support, rather than the deviation-relevant set, so it under-covers relative
to Corollary~\ref{cor:coverage}, with the exploration draws supplying the remainder. Closing
that gap with a rule that targets deviation coverage directly is an immediate direction for
future work.

\subsection{Ablations and bound diagnostics}
\label{app:ablations}

\subsubsection{Error decomposition}
Counterfactual decomposition (Figure~\ref{fig:decomp}) replaces one error
source at a time with ground truth and re-solves. Error is representational at
low budget and statistical at high, with the crossover below $20\%$ on the
informative game, between $20\%$ and $40\%$ on the arbitrary game, and near
$40\%$ on Blotto. This yields a tuning rule: below the crossover spend on
coverage, above it on revisits.

The crossover need not be measured. Coverage and revisits buy different terms of
Theorem~\ref{thm:eps} and are priced separately: coverage decides \emph{which}
cells and drives $\beta \cdot w_{\mathrm{sur}}$ to zero, while revisits decide
\emph{how tight} and shrink $2\zeta(m,\delta)$. The rule in closed form is
therefore to cover until $w_{\mathrm{sur}} = 0$, which
Proposition~\ref{prop:coverage} bounds at $n_D + n_A$ rounds and
costs $s_D + s_A$ cells per PSRO iteration under stable supports, and then to spend every remaining rollout revisiting covered
cells until $2\zeta(m,\delta)$ meets the exploitability target. The two phases
do not compete for the same cells: revisits fall only on cells coverage has
already bought.

\begin{figure}[t]
\centering
\includegraphics[width=0.72\linewidth]{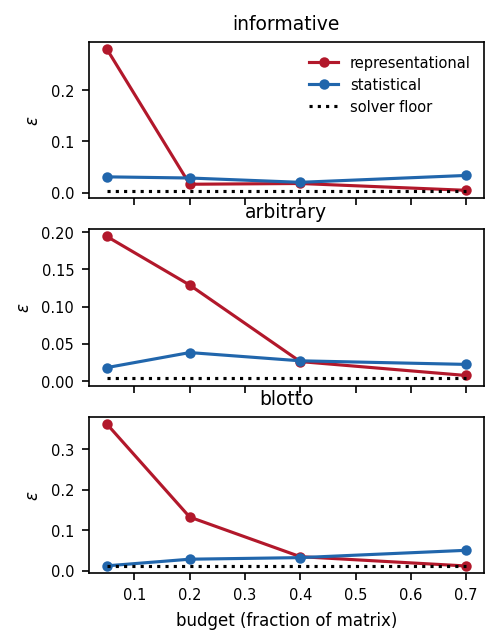}
\caption{Counterfactual decomposition of $\eps$ into statistical and
representational parts (seed means; dotted line is the solver floor).}
\label{fig:decomp}
\end{figure}

\subsubsection{Batch size and exploration probability}
\label{sec:ablation-batch-eps}

The acquisition loop has two knobs. The batch size $c$ is the number of cells
bought between refits: $c = B$ commits the whole budget on a single ranking,
$c = 1$ re-ranks after every cell. The exploration probability
$\eps_{\mathrm{x}}$ is the chance that a given selection ignores the score.
Table~\ref{tab:ablation} sweeps both at a $20\%$ budget over four seeds.

\begin{table}[t]
\centering\small
\caption{Exploitability against batch size $c$ and exploration probability
$\eps_{\mathrm{x}}$, $20\%$ budget, mean over four seeds. Batch size is the
operative knob; $\eps_{\mathrm{x}}$ is close to flat.}
\label{tab:ablation}
\begin{tabular}{@{}lrrrr@{}}
\toprule
& \multicolumn{4}{c}{$\eps_{\mathrm{x}}$} \\
\cmidrule(l){2-5}
$c$ & $0.00$ & $0.05$ & $0.15$ & $0.30$ \\
\midrule
\multicolumn{5}{@{}l}{\emph{Latent, informative}} \\
\quad $1$  & $0.018$ & $0.029$ & $0.019$ & $0.028$ \\
\quad $4$  & $0.030$ & $0.054$ & $0.056$ & $0.027$ \\
\quad $16$ & $0.014$ & $0.031$ & $0.037$ & $0.042$ \\
\quad $64$ & $0.041$ & $0.028$ & $0.033$ & $0.074$ \\
\midrule
\multicolumn{5}{@{}l}{\emph{Blotto}} \\
\quad $1$  & $0.073$ & $0.042$ & $0.078$ & $0.074$ \\
\quad $4$  & $0.086$ & $0.074$ & $0.107$ & $0.106$ \\
\quad $16$ & $0.164$ & $0.164$ & $0.194$ & $0.160$ \\
\quad $64$ & $0.163$ & $0.139$ & $0.214$ & $0.167$ \\
\bottomrule
\end{tabular}
\end{table}

Batch size is the parameter that matters, and it matters where the theory says
it should. On Blotto, $c = 1$ reaches $0.073$ against $0.163$ at $c = 64$, a
factor of $2.2$; on the informative game the effect is within seed noise. The
asymmetry follows from what a batch commits to: every cell in a batch is chosen
from one ranking, so a large $c$ spends that share of the budget on the
surrogate's state at the start of the batch. On a game the surrogate fits well
that state is already good and the loss is small; on a combinatorial surface it
is not, and re-ranking after each cell recovers most of the gap.

The exploration probability is close to flat, and $\eps_{\mathrm{x}} = 0$ is
best or tied in several cells. This is the expected reading rather than a
contradiction. Theorem~\ref{thm:trap} establishes that exploration is
\emph{necessary}, in the sense that no rule reading only fingerprints and the
cache carries the asymptotic guarantee without it; it does not claim that
exploration improves finite-budget exploitability on a benign instance. These
three games are benign in exactly that sense---no fingerprint collision is
constructed---so exploration buys the guarantee and not the number, and its cost
is the small budget share it diverts. The instance where it earns its keep is
the twin--trap construction, reported separately in
Section~\ref{sec:results}. We therefore keep $\eps_{\mathrm{x}} = 0.15$ as the
default: it is insurance whose premium is measured here and whose payout is
measured there.

\subsubsection{Variants of the acquisition score}
\label{app:acq-variants}

Table~\ref{tab:acq} compares our rule against prior work and against no score at
all. This subsection varies our own rule instead, under the same protocol:
identical surrogate, budget, batch schedule and forced-exploration rate, sixteen
paired seeds, $\eps_{12}$. None of these is a baseline drawn from the
literature; each removes or replaces one ingredient of the score in
Algorithm~\ref{alg:build}.

\begin{table}[h]
\centering\small
\caption{Variants of our acquisition score, extending Table~\ref{tab:acq} with three
further rules under the identical protocol; the first row repeats Table~\ref{tab:acq}
for reference. Mean $\pm$ one standard error over sixteen paired seeds.}
\label{tab:acq-variants}
\begin{tabular}{@{}lcc@{}}
\toprule
Acquisition score & Informative & Blotto \\
\midrule
$\nu$-smoothed product (ours)      & $\mathbf{0.024 \pm 0.004}$ & $\mathbf{0.082 \pm 0.011}$ \\
bootstrap inclusion                & $0.035 {\scriptstyle\pm} 0.006$ & $0.100 {\scriptstyle\pm} 0.017$ \\
training exposure                  & $0.039 {\scriptstyle\pm} 0.005$ & $0.105 {\scriptstyle\pm} 0.011$ \\
value-sensitivity ($\nu = 0$)      & $0.048 {\scriptstyle\pm} 0.010$ & $0.096 {\scriptstyle\pm} 0.015$ \\
\bottomrule
\end{tabular}
\end{table}

Value-sensitivity ($\nu = 0$) isolates the smoothing, and the ordering reverses
across games. It is the rule used in earlier versions of this work: no smoothing,
and the product returned inside the bootstrap average. Smoothing is worth $0.024$ on the latent games
($p = 0.007$, $12/16$ seeds) and $0.014$ on Blotto, where it is not significant
($p = 0.348$, $9/16$). It therefore does its work where the support estimate is
sharp and the unsmoothed score concentrates on a handful of cells, not where the
support is already broad---on Blotto, at supports $(12,12)$, the unsmoothed score
already spreads over most of the matrix and there is little left for $\nu$ to do.

Bootstrap inclusion replaces equilibrium mass with inclusion frequency, the
fraction of bootstrap replicates in which a strategy appears in the support at
all, and sums the two roles rather than multiplying them. This is the discrete
analogue of the probability-of-equilibrium criterion of
\citet{picheny2016bayesian}, with the bootstrap standing in for their
Gaussian-process posterior. The sum is our departure from it: a product is supported only on
$\operatorname{supp}(\hat p) \times \operatorname{supp}(\hat q)$, whereas a sum
scores every row against any plausible column, which is the shape of the
deviation-relevant set of Corollary~\ref{cor:coverage}. It is the same
neighbourhood idea as $\nu$, learned from the bootstrap rather than imposed by a
fixed constant, and it finishes second here---close enough that the choice
between them is not settled by these games.

Training exposure is the PSRO-specific rule. Every strategy in the pool arrived
as a best response to a known opponent mixture, so the cells pairing
it with the opponents it was optimised against are cells whose values that
optimisation has already partly determined; the cells it was never trained
against are both less predictable and, by Corollary~\ref{cor:coverage},
deviation-relevant as soon as that strategy enters a support. The rule
down-weights a cell by the training exposure $e(i,j)$, the weight the mixture
that produced $a_j$ placed on defender $i$, and symmetrically. It is the only
rule considered anywhere in this paper that uses information no payoff-build
baseline can see, since only the outer loop knows what each strategy was trained
against, and it finishes fourth. We report it because the negative result is
informative: a best response supplies one scalar constraint on a $\hat p$-weighted
average of a column, not per-cell information, and weighting cells by it asks the
constraint to carry more than it does. The constraint is better used inside the
surrogate fit, which we have not attempted here.

\subsection{Bound and certificate diagnostics}
\label{app:diagnostics}

These three studies use fixed pools as a diagnostic bed rather than the
growing-pool setting of Section~\ref{sec:results}. The first shows that the
deviation-relevant set predicts cost before any budget is spent. The second
breaks Figure~\ref{fig:perturb} down by game. The third gives what the computable
certificate costs in optimism. None is needed to read Table~\ref{tab:psro},
Table~\ref{tab:acq}, or Figure~\ref{fig:perturb}.

\subsubsection{Predicting difficulty before spending}
\label{sec:coverage}
The results above are retrospective: they report what a budget bought. The
practical question is the prospective one, whether a game's cost can be known
before committing to it. Corollary~\ref{cor:coverage} answers it with
$B^\ast = n_D s_A + n_A s_D - s_D s_A$, computable from bootstrap support
estimates alone, and Table~\ref{tab:coverage} evaluates it directly: simulate
exactly the deviation-relevant set, surrogate elsewhere.

Two claims are tested. First, $B^\ast$ \emph{orders the games correctly}. The
two latent games have supports of size $2$ on each side and cost $80$ cells,
$18\%$ of the matrix; Blotto's supports come out at $(12,12)$ under this
solver, giving $B^\ast = 360$ cells, $82\%$. That ordering is available before
any budget is spent, and it matches the ordering the growing-pool results of
Section~\ref{sec:psro-headline} reveal. Equilibrium support size, not matrix size, is the
operative quantity. Second, simulating exactly that set \emph{delivers
full-rebuild quality}: on the informative game coverage reaches
$0.017 \pm 0.003$ against $0.033 \pm 0.014$ for a full rebuild, at $18\%$ of the
simulation, and on the arbitrary game it is statistically indistinguishable
($0.038 \pm 0.020$). Coverage is therefore not merely equal to a full rebuild on the
informative game, it is slightly ahead of it, and the reason is worth stating
because it looks at first like an error. A full rebuild measures every one of
the $441$ cells independently at $N_{\mathrm{MC}} = 4$ and keeps the statistical
error of all of them. Coverage measures $80$ cells and predicts the rest from a
model fitted to everything observed so far. When that model class contains the
truth, its predictions pool information across the matrix and carry \emph{lower}
variance than an individual four-rollout average, so replacing measurements by
predictions can reduce error rather than add it. The latent games are the
extreme case, since their payoff matrices are exactly rank one by construction
(Section~\ref{sec:experiments}) and the surrogate's row-and-column fallback is
correctly specified for them. Applying a rank-one projection to a full rebuild's
own noisy measurements reproduces the effect on its own, halving realised
exploitability, which confirms that the gain comes from denoising rather than
from anything the acquisition rule does.

The effect tracks how well the model class contains the truth, and it disappears
when that fails. On the arbitrary game the payoffs are the same but the
fingerprints carry no information, so only row and column effects can be pooled
and coverage is indistinguishable from a rebuild. On Blotto it is worse. The
payoff surface there has effective rank above four, the surrogate cannot fit it,
and at $82\%$ coverage there is little left to save in any case.

\begin{table}[t]
\centering
\small
\caption{Deviation-relevant coverage (Corollary~\ref{cor:coverage}), mean
$\pm$ SD over sixteen seeds, on $21 \times 21$ games ($441$ cells).}
\label{tab:coverage}
\begin{tabular}{lccc}
\toprule
\textbf{Game} & \textbf{supp.} & $B^\ast$ \textbf{(\%)} & $\eps$ \\
\midrule
Latent, informative & $(2,2)$   & $80$ \ ($18$) & $0.017 \pm 0.003$ \\
Latent, arbitrary   & $(2,2)$   & $80$ \ ($18$) & $0.038 \pm 0.020$ \\
Blotto              & $(12,12)$ & $360$ ($82$)  & $0.062 \pm 0.034$ \\
\midrule
\multicolumn{2}{l}{\emph{full rebuild, latent}} & $441$ ($100$) & $0.033 \pm 0.014$ \\
\multicolumn{2}{l}{\emph{full rebuild, Blotto}} & $441$ ($100$) & $0.043 \pm 0.021$ \\
\bottomrule
\end{tabular}
\end{table}

\subsubsection{Bound tightness by game}

\begin{figure}[t]
\centering
\includegraphics[width=0.72\linewidth]{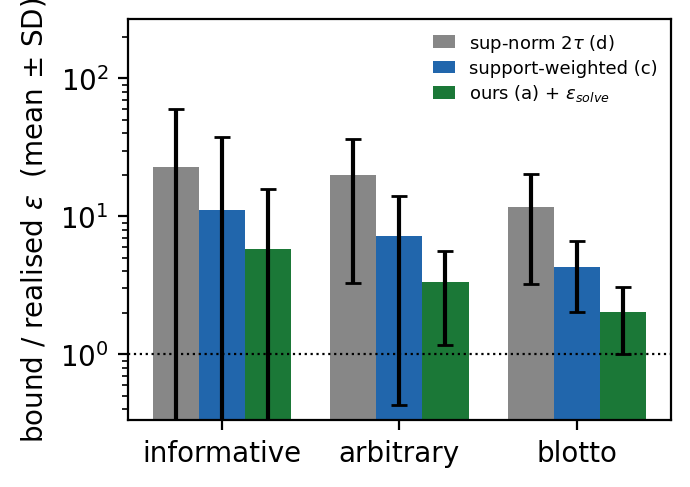}
\caption{Overestimation factor, bound over realised $\eps$, per game (mean
$\pm$ SD over the $240$ budget-sweep runs, log scale). Labels refer to the
forms of Eq.~\eqref{eq:chain}; our bound is $4$--$6\times$ tighter than the
sup-norm bound.}
\label{fig:bounds}
\end{figure}

Figure~\ref{fig:bounds} reports the same $240$ runs as per-game overestimation
factors. The sup-norm bound (d) overestimates realised $\eps$ by mean factors of
$22.7$, $20.0$, and $11.7$ on the three games; the support-weighted form (c)
reduces these to $11.1$, $7.2$, and $4.3$; the signed form (a) with the measured
solver term to $5.8$, $3.4$, and $2.0$. Mean ratios carry heavy right tails,
since a run with near-zero realised $\eps$ inflates the ratio however tight the
bound is in absolute terms, so the seed spreads are large. The ordering
(d) $>$ (c) $>$ (a) is uniform across games and budgets.

\subsubsection{The computable certificate, and the cost of optimism}
\label{sec:cert}
Corollary~\ref{cor:cert} is evaluated in both of the roles
Section~\ref{sec:theory} distinguishes. Used post hoc at $B^\ast$ coverage,
the computable certificate held on all $48$ runs, with values of
$0.935 \pm 0.059$ and $0.938 \pm 0.075$ on the latent games and
$0.516 \pm 0.032$ on Blotto against realised $\eps$ of $0.017$--$0.062$
(Table~\ref{tab:coverage}). At $m = 4$ visits its slack is dominated by the
time-uniform concentration radius $2\zeta_{4,0.05} = 0.481$ and by mismatch
between the estimated and true supports, and it contracts as $1/\sqrt{m}$ with
revisits. Proposition~\ref{prop:coverage} removes the second of those two
sources by construction: iterating coverage to $w_{\mathrm{sur}} = 0$ leaves the
statistical radius as the only slack.

That remaining slack is the operative constraint, and it is a visit count rather
than a coverage question. At $m = 4$ the bound a fully covered build can certify
is $\eps_{\mathrm{solve}} + 0.481$, roughly half the payoff range, which is why
the certificates above sit near $0.94$ against realised $\eps$ of $0.017$--$0.062$:
the estimator was buying the right cells and not enough visits. Inverting
$2\zeta(m,\delta)$ prices the alternative. A certificate of $0.25$ needs $m = 16$
and, at the $80$ cells of $B^\ast$ coverage, $1{,}280$ rollouts---below the
$1{,}764$ a full rebuild spends at $m = 4$. A certificate of $0.10$ needs
$m = 101$ and $8{,}080$ rollouts, and one of $0.05$ needs $m = 408$ and
$32{,}640$, some $18\times$ a full rebuild. Coverage is cheap and the certificate
is not; a deployment should choose the bound it needs and let the visit count
follow, rather than choosing a cell budget and discovering what it certifies. Used inside the
solve, R-max style, optimism certified only after simulating $90$--$97\%$ of
the matrix ($395$--$430$ of $441$ cells), because every unsimulated cell looks
maximal and the support chases fresh cells. The division of labour is
therefore a three-way one: the surrogate fill decides which equilibrium is
solved for, coverage of the deviation-relevant set makes the bound on that
equilibrium independent of the fill, and optimism only certifies it. The
$B^\ast$ runs above isolate the middle role---every one of them carries the same
$\beta$-free guarantee---while the exploitability they reach still separates
surrogate fill from the flat fill of the uniform baseline, which is the first
role.

\end{document}